\documentclass[11pt]{article}

\usepackage[letterpaper]{geometry}
\usepackage{fullpage}
\usepackage{amsmath,amsthm,amssymb}
\usepackage{graphicx}
\usepackage{algorithmic}
\usepackage{enumerate}
\usepackage{latexsym}
\usepackage{caption}
\usepackage{multirow}
\usepackage{paralist}
\usepackage[dvipsnames]{xcolor}
\usepackage[normalem]{ulem}
\usepackage{authblk}
\usepackage{soul}
\usepackage[normalem]{ulem}

\usepackage{booktabs} 
\usepackage[noend,ruled,noline,linesnumbered]{algorithm2e} 

\usepackage{enumitem}
\usepackage{mathtools}
\usepackage{amsmath,amsthm}

\SetAlFnt{\small}
\SetAlCapFnt{\small}
\SetAlCapNameFnt{\small}
\SetAlCapHSkip{0pt}

\IncMargin{-\parindent}

\usepackage[numbers]{natbib}

 \newtheorem{theorem}{Theorem}
 
  \newtheorem{lemma}[theorem]{Lemma}
  
 \newtheorem{definition}[theorem]{Definition}
  \newtheorem*{definition*}{Definition}
  
    \newtheorem{observation}[theorem]{Observation}

  \theoremstyle{remark}
 
   \usepackage{thm-restate}

\usepackage{tikz}
\usepackage{verbatim}
\usetikzlibrary{shapes,arrows,fit,calc,positioning}
\usetikzlibrary{arrows}
\usetikzlibrary{decorations.markings}

\newcommand\numberthis{\addtocounter{equation}{1}\tag{\theequation}}

\renewcommand{\geq}{\ensuremath{\geqslant}}
\renewcommand{\leq}{\ensuremath{\leqslant}}

\newcommand{\err}{\eta}

\usepackage{soul}

\DeclareMathOperator*{\argmax}{arg\,max}
\DeclareMathOperator*{\argmin}{arg\,min}

\newcommand{\poptm}{\hat{\imath}}

\newcommand{\optm}{i^\star}

\newcommand{\mech}{\textsc{ScaledGreedy}}
\newcommand{\simplemech}{\textsc{SimpleScaledGreedy}}
\newcommand{\mecherr}{\textsc{ErrorTolerantScaledGreedy}}
\newcommand{\errbound}{\bar{\eta}}
\newcommand{\algo}{\textsc{Makespan-Min}}

\newcommand{\OPT}{\texttt{OPT}}

\newcommand{\ms}{\texttt{MS}}
\newcommand{\instance}{\mathbf{p}}
\newcommand{\schedule}{\mathbf{x}}
\newcommand{\tradeoff}{\gamma}

\allowdisplaybreaks

\title{Strategyproof Scheduling with Predictions}

\author[a]{Eric Balkanski\thanks{eb3224@columbia.edu}}
\author[b]{Vasilis Gkatzelis\thanks{gkatz@drexel.edu}}
\author[b]{Xizhi Tan\thanks{xizhi@drexel.edu}}
\affil[a]{Columbia University, IEOR}
\affil[b]{Drexel University, Computer Science}

\begin{document}

\date{}
\maketitle

\begin{abstract}
In their seminal paper that initiated the field  of algorithmic mechanism design, \citet{NR99} studied the problem of designing strategyproof mechanisms for scheduling jobs on unrelated machines aiming to minimize the makespan. They provided a strategyproof mechanism that achieves an $n$-approximation and they made the bold conjecture that this is the best approximation achievable by any deterministic strategyproof scheduling mechanism. After more than two decades and several efforts, $n$ remains the best known approximation and very recent work by  \citet{CKK21} has been able to prove an $\Omega(\sqrt{n})$ approximation lower bound for all deterministic strategyproof mechanisms. This strong negative result, however, heavily depends on the fact that the performance of these mechanisms is evaluated using worst-case analysis. To overcome such overly pessimistic, and often uninformative, worst-case bounds, a surge of recent work has focused on the ``learning-augmented framework'', whose goal is  to leverage machine-learned  predictions to obtain improved approximations when these predictions are accurate (consistency), while also achieving near-optimal worst-case approximations even when the predictions are arbitrarily wrong (robustness).

In this work, we study the classic strategic scheduling problem of~\citet{NR99} using the learning-augmented framework and give a deterministic polynomial-time strategyproof mechanism that is $6$-consistent and $2n$-robust. We thus achieve the ``best of both worlds'': an $O(1)$ consistency and an $O(n)$ robustness that asymptotically matches the best-known approximation. We then extend this result to provide more general worst-case approximation guarantees as a function of the prediction error. Finally, we complement our positive results by showing that any $1$-consistent deterministic strategyproof mechanism has unbounded robustness.
\end{abstract}

\thispagestyle{empty}
\newpage
\setcounter{page}{1}

\section{Introduction}
In their seminal paper which initiated the field of algorithmic mechanism design, \citet{NR99} focused on a natural job scheduling problem involving strategic agents: a set of $m$ jobs needs to be scheduled on a set of $n$ machines aiming to minimize the makespan, and each machine is owned by an agent who requires a monetary compensation in exchange for processing the jobs assigned to them. What makes this problem particularly demanding is that the compensation each agent receives needs to be at least as high as the cost that they suffer for processing the jobs assigned to them (i.e., the required processing time), yet the actual processing times are private information that only that agent knows. Therefore, an agent can choose to misreport (either overstate or understate) their processing times if doing so would change the outcome (i.e., their assignment and their compensation) to one that they prefer. To avoid this issue, the goal is to design \emph{strategyproof} mechanisms that carefully choose the outcome to ensure that no agent has an incentive to misreport their costs, while also ensuring that the makespan of the chosen assignment is as small as possible.

\citet{NR99} proposed a strategyproof mechanism for this problem and showed that the allocation generated by this mechanism is an $n$-approximation of the optimal makespan. Even though this guarantee is much less appealing than the $2$-approximation that one can achieve in polynomial time if the costs are publicly known~\citep{LST90}, they went on to make the bold conjecture that this is the best possible worst-case approximation that any deterministic strategyproof mechanism can achieve (polynomial time or not). After more than two decades, and despite several efforts to tackle this, now classic, problem, no strategyproof mechanism with an approximation better than $O(n)$ is known (not even randomized). Instead, after a sequence of inapproximability results, proving gradually increasing constant approximation lower bounds for deterministic mechanisms (e.g., \cite{NR99,CKV09,KV07,GH20,DS20}), a very recent breakthrough by \citet{CKK21} made major progress by proving a lower bound of $1+\sqrt{n-1}$ for deterministic mechanisms. This result paints a very pessimistic picture, implying that even if a better strategyproof mechanism exists, its worst-case approximation will not be practical. However, this heavily depends on the, rather unrealistic assumption, that the mechanism has no information regarding the costs of the agents.

Facing analogous worst-case impossibility results that are overly pessimistic and rather uninformative, a recent surge of work has focused on the ``learning-augmented framework'', aiming to achieve more refined and informative bounds, while retaining the benefits of worst-case analysis (see  \citep{MV21} for a survey). This framework assumes that the designer is provided with some machine-learned predictions that can be used to design more practical algorithms that achieve near optimal performance when the predictions are correct (this is called the \emph{consistency} guarantee). However, rather than assuming that the predictions are always accurate and forfeiting the benefits of worst-case analysis altogether, this framework seeks to also provide strong guarantees, even if the predictions are arbitrarily inaccurate (this is called the \emph{robustness} guarantee). More formally, the consistency of an algorithm is the worst-case approximation guarantee that it achieves assuming that the prediction it was provided with is accurate, while its robustness is the worst-case approximation guarantee without any assumptions regarding the prediction accuracy (see Section~\ref{sec:prelims} for more details). The quality of an algorithm is then evaluated based on both its consistency and its robustness. This framework was initially restricted to algorithms, but it was very recently extended to settings involving strategic agents \citep{ABGOT22,GKST22,XL22}, motivating the design of new mechanisms leveraging predictions to overcome the incentive issues that arise. In this paper, we revisit the classic scheduling problem of \citet{NR99} using the learning-augmented framework and design mechanisms enhanced with predictions regarding the machines' processing times.

Given predictions regarding the machines' processing times, a naive solution would be to assume the predictions are accurate and output a schedule with small makespan according to the predicted processing times alone (i.e., without eliciting the agents' true processing times). This is a strategyproof mechanism that would achieve a consistency of $O(1)$, but its robustness would be unbounded (e.g., even if a single job's processing time is mispredicted, it could end up being assigned to a machine where its actual processing time is arbitrarily high). On the other hand, the mechanism of \citet{NR99} would achieve a robustness of $O(n)$, but its consistency would also be no better than $O(n)$, since its output disregards the predictions. In an attempt to provide a middle ground between these two extremes, very recent work by \citet{XL22} proposed a strategyproof mechanism that can guarantee a consistency of $O(1)$ with bounded robustness, but the robustness guarantee of $O(n^3)$ that it achieves is much worse than the $O(n)$ robustness achieved in \cite{NR99}. Our main result in this paper is a strategyproof that combines the ``best of both worlds'': a consistency of $O(1)$ with the best-known robustness of $O(n)$.

\subsection{Our results}
In this paper we leverage the learning-augmented framework to develop a deeper understanding of this classic strategic scheduling problem, and design much more practical mechanisms enhanced with predictions. In particular, we assume that the mechanism is provided with predictions $\hat{p}(i,j)$ regarding the amount of time that each machine $i$ would require in order to fully process each job $j$; crucially, these predictions may be highly inaccurate. 

Our main result is a deterministic polynomial-time strategyproof mechanism that guarantees a consistency of $(4+2\tradeoff)$ and a robustness of $(1 + \frac{1}{\tradeoff})n$ for any choice of $\tradeoff \in \left(0, \frac{n}{2}-1\right)$, which is a parameter that the designer can determine to receive their desired trade-off between consistency and robustness. For example, setting $\gamma=1$ yields a mechanism with a small constant consistency of $6$ (i.e., a $6$-approximation of the optimal makespan when the predictions are accurate) while maintaining a robustness of $2n$ (i.e., a worst-case approximation guarantee that asymptotically matches the best-known approximation, irrespective of how inaccurate the predictions may be). 

Our main mechanism, \mech, uses the predicted processing times to pre-compute a schedule that (approximately)\footnote{Note that even if the processing times are known in advance, no known polynomial time algorithm can achieve an approximation factor better than 2, and it is NP-hard to achieve a factor better than $1.5$ \citep{LST90}.} minimizes the makespan, assuming these predictions are correct. The mechanism then uses this schedule as a guide and adds a ``bias'' in favor of scheduling jobs on the machines that they were assigned to in this schedule. After introducing this bias, the mechanism follows a simple greedy procedure, so the main technical novelty is the subtle way in which this bias towards the predictions is introduced, while ensuring that we asymptotically match the best known robustness when the predictions are arbitrarily inaccurate. Roughly speaking, the mechanism carefully chooses a subset of jobs and proceeds to assign each of them to the same machine as in the pre-computed schedule, unless the processing times reported by the machines for that job differ significantly from the predicted ones.
Before presenting this mechanism and its analysis in Section~\ref{sec:warmup}, we first dedicate Section~\ref{sec:scaledGreedy} to a simpler mechanism, \simplemech, which provides a warm-up toward our main result while combining a consistency of $O(1)$ with a robustness of $O(n^2)$. 

In Section~\ref{sec:error}, we take one step further and propose a mechanism that provides worst-case approximation guarantees as a function of the prediction error. In other words, rather than focusing only on the two extremes of consistency (where the error is zero) and robustness (where the error is unbounded), this mechanism guarantees a good approximation as long as the prediction error is not too high. Given a prediction $\hat{\instance}$ of the actual processing times $\instance$, we let the prediction error $\eta$ be the largest ratio between the predicted processing time and actual processing time for any machine $i,$ job $j$ pair, i.e., $\eta = \max_{i, j} \max\left\{\frac{\hat{p}(i, j)}{p(i,j)}, \frac{p(i,j)}{\hat{p}(i,j)}\right\}.$ Our mechanism takes as input an error tolerance parameter $\errbound > 0$ and achieves an approximation of  $O(\eta^2)$ as long as $\eta \leq \errbound$ and $O(\errbound^2 n)$ otherwise. 

Finally, in Section~\ref{sec:lower_bound} we complement our positive results with an impossibility result, showing that achieving a consistency of $1$, i.e., returning an optimal schedule whenever the predictions are accurate, necessary leads to unbounded robustness (irrespective of any computational limitations).

\subsection{Related work}
\paragraph{Strategyproof scheduling}
Apart from proposing their $n$-approximate mechanism and conjecturing that this is the best possible guarantee across all deterministic and strategyproof mechanisms, \citet{NR99, NR01} also proved an approximation lower bound of $2$. This lower bound was later improved to 2.41 by \citet{CKV09}, then 2.61 by \citet{KV07}. More recently, \citet{GH20} further raised this lower bound to 2.755 and then \citet{DS20} pushed it to 3. Whether a constant approximation is possible remained an open problem until very recently, when the breakthrough result of \cite{CKK21} eventually proved a lower bound of $1 + \sqrt{n-1}$. Earlier work by \citet{ADL12} had also shown a lower bound of $n$ for the special class of mechanisms that are anonymous. Note that although these lower bounds focus on deterministic mechanisms, even if we allow randomization, the best known approximation guarantee achievable by a randomized strategyproof mechanism remains $O(n)$~\cite{CKK10}.  

\vspace{-.2cm}

\paragraph{Learning-augmented framework.}
Motivated by the shortcomings of worst-case analysis, an exciting new literature has focused on the design and analysis of ``learning-augmented algorithms'', or ``algorithms with predictions'' (see \citep{MV21} for a survey of some early work and \cite{alps} for a more up-to-date list of papers in this area). This literature assumes that the algorithm is provided with predictions regarding the instance at hand and its performance is evaluated using \emph{consistency} and \emph{robustness}, which are the two primary metrics introduced by \citet{lykouris2018competitive}. During the last five years, more than 100 papers have revisited classic algorithmic problems using this framework, including online paging \citep{lykouris2018competitive},  scheduling \citep{KPZ18}, and secretary problems \citep{dutting2021secretaries,AGKK20}, optimization problems with covering \citep{BMS20} and knapsack constraints \citep{im2021online}, as well as Nash social welfare maximization \citep{banerjee2020online} and several graph problems \citep{azar2022online} . This line of work also studied online scheduling problems (e.g., \cite{PSK18,Mit,LLMV20,BMRS20,LX21,BDKLP22,BOSW22}), but it was restricted to non-strategic settings and the predictions were used to overcome information limitations regarding the future, rather than limitations regarding privately held information. Very recently, \citet{ABGOT22} and, independently, also~\citet{XL22} extended the framework to settings involving strategic agents. 
\citet{ABGOT22} provided optimal learning-augmented mechanisms, focusing on the strategic facility location problem, while \citet{XL22} considered a variety of different problems. One of these problems was the strategic scheduling problem that we study in this paper, for which they gave a strategyproof mechanism that is $O(\gamma)$-consistent and $O(n^3/\gamma^2)$-robust, for any $\gamma \in [1, n].$ However, to achieve any interesting consistency bounds (i.e., $o(n)$), their mechanism cannot guarantee a robustness better than $\omega(n^2)$; in particular, to achieve an optimal consistency of $O(1)$, the mechanism's robustness is $\Omega(n^3)$. Our main result significantly improves this result by combining the best of both worlds: an optimal consistency of $O(1)$ and the best-known robustness of $O(n)$.

\section{Preliminaries}\label{sec:prelims}
\paragraph{Makespan minimization on unrelated machines.} In the classic problem of scheduling on unrelated machines there is a set $N$ of $n$ machines, a set $M$ of $m$ jobs, and processing times $p(i,j)$ for a job $j \in M$ to be processed on machine $i \in N$. We use $\mathbf{x} \in \{0,1\}^{n \times m}$ to denote the \emph{allocation} of jobs to machines, where $x(i,j) =1$ if and only if job $j$ is allocated to machine $i$. Given an instance $\instance$ of processing times and an allocation $\schedule$, the \emph{makespan} is the maximum over all machines of the total processing time assigned to that machine, i.e., $\ms(\instance,\schedule) = \max_{i}\sum_{j}p(i,j)\cdot x(i,j)$. The goal is to find an allocation $\schedule$ such that the makespan is minimized. We denote the optimal makespan for an instance $\mathbf{p}$ by $\OPT(\mathbf{p}) = \min_{\mathbf{x}} \ms(\mathbf{p}, \mathbf{x})$ and we let $\mathbf{x}^\star(\mathbf{p})$ denote an optimal allocation for instance $\mathbf{p}$ (so $\ms(\mathbf{p}, \mathbf{x}^\star(\mathbf{p})) = \OPT(\mathbf{p}))$. We abuse notation with $\OPT$ and $\mathbf{x}^\star$ when the instance $\mathbf{p}$ is clear from context. 

\paragraph{Strategic makespan minimization.} In the strategic version of this problem, each machine $i$ is controlled by a distinct self-interested agent. Each such agent incurs a cost for processing jobs that is equal to its load, i.e., the sum of the processing times of the jobs assigned to $i$ and needs to receive a payment in exchange. Specifically, given an allocation $\schedule$, if agent $i$ receives a payment $\rho_i$, and their utility is equal to $\rho_i - \sum_{j : x(i,j) = 1} p(i,j)$. The processing times $\mathbf{p}_i = (p(i, 1), \ldots, p(i, m))$ are private information that is known only to the agent controlling machine $i$, and each agent can misreport this information aiming to increase her utility. A scheduling mechanism consists of an allocation rule $\mathbf{x}(\mathbf{p}) \in \{0,1\}^{n \times m}$ for assigning jobs to machines and a payment rule $\boldsymbol{\rho}(\mathbf{p}) \in \mathbb{R}^n$ for compensating the machines for their costs. A scheduling mechanism is \emph{strategyproof} if truthfully reporting $\mathbf{p}_i$ is a dominant strategy for every machine $i$, i.e., it is always an optimal strategy for agent $i$ to report the truth, regardless of what every other agent reports.

\paragraph{Characterization of strategyproof mechanisms.} Seminal work on strategic scheduling has characterized the family of allocations rules $\mathbf{x}(\mathbf{p})$ that admit a payment rule $\boldsymbol{\rho}(\mathbf{p})$ such that $(\mathbf{x}(\mathbf{p}), \boldsymbol{\rho}(\mathbf{p}))$ is a strategyproof scheduling mechanism.

\begin{definition}[Monotonicity Property]
An allocation algorithm is monotone if fixing the processing time of all other agents $\instance_{-i}$, for every two reports of agent $i$, $\instance_i$ and $\instance_i'$, the associated allocations $\schedule$ and $\schedule'$ satisfy
\[\sum_{j \in M} (x(i,j) - x'(i,j))(p(i,j) - p'(i,j)) \leq 0.\]
\end{definition}
\begin{lemma}[\citep{NR01}, \citep{SY05}]\label{lem:monotonicity}
An allocation algorithm $\mathbf{x}(\mathbf{p})$ admits a payment rule $\boldsymbol{\rho}(\mathbf{p})$ s.t. $(\mathbf{x}(\mathbf{p}), \boldsymbol{\rho}(\mathbf{p}))$ is  a strategyproof scheduling mechanism if and only if $\mathbf{x}(\mathbf{p})$ satisfies the monotonicity property.
\end{lemma}
This fundamental result reduces the mechanism design problem of strategic scheduling to an algorithm design problem where the goal is to find a (near)-optimal monotone allocation. One simple but natural mechanism that satisfies the above property is the \emph{greedy mechanism}: it assigns each job $j$ to the machine with minimum processing time for job $j$, i.e., $x(i,j) = 1$ if $i = \argmin_{i'\in N}p(i',j)$. \citet{NR01} analyzed this mechanism and showed that it achieves an $n$-approximation to the optimal makespan.

\paragraph{Learning-augmented mechanism design.} In the learning-augmented mechanism design framework, before requesting the processing time vector $\instance_i$ from each machine $i$, the designer is provided with a prediction $\hat{\instance}$ of the entire $n \times m$ processing time values $p(i,j)$. The designer can use this information to choose the rules of the mechanism but, as in the standard stretagyproof scheduling setting, the final mechanism, denoted $\schedule(\instance,\hat{\instance})$, needs to be strategyproof. We use \emph{consistency} and \emph{robustness} to measure the performance of the mechanism. where consistency is the worst-case approximation given a accurate prediction, i.e., $\instance = \hat{\instance}$ and the robustness is the worst-case approximation given any predictions. Given a instance $\instance$, a mechanism is $\alpha$-\emph{consistent} if it achieves an $\alpha$-approximation ration when the prediction is correct ($\hat{\instance} = \instance$), i.e., 
\[\max_{\instance} \left\{ \frac{\ms(\instance,\schedule(\instance, \instance))}{\ms(\instance, \schedule^\star(\instance))}\right\} \leq \alpha.\]
A mechanism is $\beta$-\emph{robust} if it achieves a $\beta$-approximation ratio even when the predictions are arbitrarily wrong, i.e.,
\[\max_{\instance, \hat{\instance}} \left\{ \frac{\ms(\instance,\schedule(\instance, \hat{\instance}))}{\ms(\instance, \schedule^\star(\instance))}\right\} \leq \beta.\]
One can also evaluate the worst-case approximation ratio as a function of the prediction error $\eta \geq 0$. In this paper we define the error $\eta(\instance, \hat{\instance})$ as the largest ratio between the predicted processing time and actual processing time for any $i$ $j$ pair, formally,
$\eta = \max_{i \in N, j \in M} \max\left\{\frac{\hat{p}(i,j)}{p(i,j)}, \frac{p(i,j)}{\hat{p}(i,j)}\right\}.$
Given a prediction error $\eta$, a mechanism achieves a $\gamma(\eta)-$approximation if
\[\max_{\instance, \hat{\instance}: \eta(\instance, \hat{\instance}) \leq \eta} \left\{ \frac{\ms(\instance,\schedule(\instance, \hat{\instance}))}{\ms(\instance, \schedule^\star(\instance))}\right\} \leq \gamma(\eta).\]
Note that for $\eta = 1$ the bound corresponds to the consistency guarantee and for $\eta \rightarrow \infty$ it captures the robustness guarantee.

\section{Warm-Up: a $4$-Consistent and $n^2$-Robust Mechanism}\label{sec:warmup}
As a warm-up, we first present a strategyproof mechanism that is $4$-consistent and $n^2$-robust. In the next section, we build on this mechanism   to obtain a mechanism with $O(n)$ robustness.

\paragraph{The mechanism.} \citet{NR99} showed that the greedy mechanism that assigns each job $j$ to the machine with minimum processing time for job $j$ has a wost-case approximation factor of $n$, implying that it is $n$-consistent and $n$-robust. On the other hand, returning a schedule that is optimal for the predicted instance gives a $1$-consistent mechanism with unbounded robustness. \simplemech, described formally in Mechanism~\ref{mech:simplescaledgreedy}, is a variant of the greedy mechanism where the processing times are scaled as a function of the predictions. More formally, instead of greedily assigning job $j$ to the machine $i$ with minimum processing time $p(i,j)$, \simplemech \ assigns  $j$ to the machine $i$ with minimum scaled processing time $r(i,j) \cdot  p(i,j)$, where $r(i,j)$ are scalars that are defined by the mechanism. We note that even when this scaling depends arbitrarily on the predicted instance $\hat{\instance}$, this mechanism is monotone (which we formally show in Lemma~\ref{lem:wstrategyproof}), and thus strategyproof (these scalars are independent of the reported instance $\instance$).

The central part of our mechanism thus consists of constructing  scalars $r(i,j)$ as a function of the predictions. The mechanism first computes an assignment $\hat{\schedule}$ for the predicted instance by using a, not necessarily strategyproof, scheduling algorithm \algo \ that is given as input to the mechanism. We call $\hat{\schedule}$ the predicted assignment. Let $\poptm_j$ be the machine that $j$ is assigned to according to the predicted assignment, which we call the predicted machine of job $j$.   Note that with scalars $r(i,j) =  \frac{\hat{p}(\poptm_j,j)}{\hat{p}(i,j)}$, we have that $r(i,j) \cdot  \hat{p}(i,j) = \hat{p}(\poptm_j,j)$ for all jobs $j$. Thus, with such scaling factors and by breaking ties in favor of $\poptm_j$, each job $j$ would be assigned to its predicted machine $\poptm_j$. The issue with such scaling factors is that they can be either arbitrarily large or arbitrarily small, which causes the robustness to be unbounded since these factors can cause a job to be allocated on an undesirable machine when the predictions are inaccurate. To avoid this issue, we cap the scaling factors to ensure that they have values between $1$ and $n$. We break ties in favor of $\poptm_j$ when $r(\poptm_j,j) \cdot  p(\poptm_j,j) = \argmin_i  r(i,j) \cdot  p(i,j)$, and otherwise we break them in favor of the machine with minimum index.

Note that when the prediction is correct, i.e., $\hat{\instance} = \instance$, the mechanism assigns a job $j$ to its predicted machine $\poptm_j$ unless  $\min_i \hat{p}(i,j) \leq \frac{\hat{p}(\poptm_j,j)}{n}$, i.e., unless there exists a machine $i$ with much smaller predicted processing time for $j$ than the predicted machine $\poptm_j$ of job $j$. If   $\min_i \hat{p}(i,j) \leq \frac{\hat{p}(\poptm_j,j)}{n}$, and the prediction is correct, the mechanism  assigns $j$ to the machine $\argmin_i p(i,j)$ with the smallest true (and predicted) processing time for $j$. 

\begin{algorithm}[h]
\textbf{Input:} instance $\instance \in \mathbb{R}^{n \times m}$, predicted instance $\hat{\instance} \in \mathbb{R}^{n \times m}$, scheduling algorithm \algo

$\hat{\schedule} \leftarrow \algo(\hat{\instance})$\\
$\poptm_j \leftarrow  $ the machine $i$ that job $j$ is assigned  to according to $\hat{\schedule}$, i.e., $\hat{x}(\poptm_j,j) = 1$, for each $j \in M$ \\
$r(i,j) \leftarrow \max\left(1, \min\left( \frac{\hat{p}(\poptm_j,j)}{\hat{p}(i,j)},n\right)\right)$, for each $(i,j) \in N \times M$ \\
$i_j \leftarrow \argmin_i  r(i,j) \cdot  p(i,j) $, break ties in favor of $\poptm_j$, for each $j \in M$ \label{mech1:tie} \\
\textbf{if} $i = i_j$ \textbf{then} $x(i, j) = 1$ \textbf{else} $x(i,j) = 0$, for each $(i,j) \in N \times M$ \\
\Return{$\mathbf{x}$}
\caption{\simplemech}
\label{mech:simplescaledgreedy}
\end{algorithm}

We first analyze the mechanism's consistency, then its robustness, and finally show that it is strategyproof.

\paragraph{The  consistency.} We first introduce some notation. We assume that $\algo$   is an $\alpha$-approximation algorithm for makespan minimization and let $M_i = \{j \in M: x(i,j)  = 1\}$ and $\hat{M}_i = \{j \in M:  \hat{x}(i, j) = 1\}$ be the sets of jobs $j$  such that $j$ is assigned to machine $i$ by \simplemech \ and the predicted assignment, respectively.  To bound the total processing time of the jobs $M_i$ assigned to machine $i$, for each $i \in N$,  we separately bound the total processing times of  $M_i \cap \hat{M}_i$ and $M_i \setminus \hat{M}_i$.  Bounding the total processing time of jobs in $M_i \cap \hat{M}_i$, i.e.,   jobs assigned to machine $i$ according to both the mechanism's assignment $\schedule$ and the predicted assignment $\hat{\schedule}$,  by $\alpha \OPT$ is almost immediate when the predictions are accurate.

\begin{lemma}
\label{lem:boundcap}
Assume that $\hat{\instance} = \instance$, then, for all $i \in N$, $\sum_{j \in M_i \cap \hat{M}_i} p(i,j) \leq \alpha \OPT(\instance)$.
\end{lemma}
\begin{proof}
For any machine $i \in N,$
observe that 
$$\sum_{j \in M_i \cap \hat{M}_i} p(i,j)   \leq \sum_{j \in \hat{M}_i} p(i,j) 
 = \sum_{j \in \hat{M}_i} \hat{p}(i,j) \leq \ms(\hat{\instance}, \hat{\schedule}) \leq \alpha \OPT(\hat{\instance}) = \alpha \OPT(\instance),$$
where first equality is since $\instance = \hat{\instance}$, the second inequality is by definition of $\hat{M}_i$ and the makespan of allocation $\hat{\schedule}$, the last inequality is since $\algo$  is an $\alpha$-approximation algorithm, and the last equality is  since $\hat{\instance} = \instance$.
\end{proof}

The main part of the analysis of the consistency  bound is the total processing time of jobs that are assigned to $i$ even though $i$ is not their predicted assignment, i.e., the total processing time of $M_i \setminus \hat{M}_i$. We first show that such jobs must have a scalar $r(i,j) = n$.

\begin{lemma}
\label{lem:scalefactorn}
For any job $j \in M$ and machine $i \in N$, if $\hat{\instance} = \instance$ and $j  \in M_i \setminus \hat{M}_i$, then $r(i,j) = n$.
\end{lemma}
\begin{proof}
Consider a job $j \in M$ and machine $i \in N$ such that $\hat{\instance} = \instance$ and $j  \in M_i \setminus \hat{M}_i$. Since $j \in M_i \setminus \hat{M}_i$, we have that $\poptm_j \neq i$ and  $r(i,j) \cdot  p(i,j) < r(\poptm_j,j) \cdot  p(\poptm_j,j)$ by definition of \simplemech. We get that
$$r(i,j) < r(\poptm_j,j) \cdot  \frac{p(\poptm_j,j)}{p(i,j)}  = \max\left(1, \min\left( \frac{\hat{p}(\poptm_j,j)}{\hat{p}(\poptm_j,j)},n\right)\right)  \cdot  \frac{p(\poptm_j,j)}{p(i,j)} = \frac{p(\poptm_j,j)}{p(i,j)} = \frac{\hat{p}(\poptm_j,j)}{\hat{p}(i,j)},$$
where the first equality is by definition of $r(i,j)$ and the last equality since 
$\hat{\instance} = \instance$. Since $r(i,j) < \frac{\hat{p}(\poptm_j,j)}{\hat{p}(i,j)}$
and, by definition, $r(i,j) = \max\left(1, \min\left( \frac{\hat{p}(\poptm_j,j)}{\hat{p}(i,j)},n\right)\right) $, we obtain that $r(i,j) = n$.
\end{proof}

We are now ready to bound the total processing time of $M_i \setminus \hat{M}_i$.

\begin{lemma}
\label{lem:boundintersect}
Assume that $\hat{\instance} = \instance$, then, for all $i \in N$, $\sum_{j \in M_i \setminus \hat{M}_i} p(i,j) \leq  \alpha \cdot \OPT(\instance)$.
\end{lemma}
\begin{proof}
Let $i$ be an arbitrary machine and assume that $\hat{\instance} = \instance$. First, observe that
$$\sum_{j \in M_i \setminus \hat{M}_i} p(i,j) \leq \sum_{j \in M_i \setminus \hat{M}_i} \frac{r(\poptm_j,j)}{r(i,j)} p(\poptm_j, j) = \frac{1}{n}\sum_{j \in M_i \setminus \hat{M}_i} p(\poptm_j, j) \leq \frac{1}{n}\sum_{j \in M} p(\poptm_j, j) = \frac{1}{n}\sum_{j \in M} \hat{p}(\poptm_j, j)$$
where the first inequality is since $j$ was assigned to machine $i$ by \simplemech, the first equality since $r(\poptm_j,j) = 1$ and  $r(i,j) = n$ for $j \in M_i \setminus \hat{M}_i$ by Lemma~\ref{lem:scalefactorn}. We also have that 
$$\frac{1}{n}\sum_{j \in M} \hat{p}(\poptm_j,j) = \frac{1}{n} \sum_{i \in N} \sum_{j \in \hat{M}_i} \hat{p}(i,j) \leq \max_{i \in N} \sum_{j \in \hat{M}_i} \hat{p}(i,j) =   \ms(\hat{\instance}, \hat{\schedule}) \leq \alpha  \cdot \OPT(\hat{\instance}) = \alpha \cdot \OPT(\instance)$$
where the last inequality is since $\algo$  is an $\alpha$-approximation algorithm. Combing the above two series of inequalities, we get $\sum_{j \in M_i \setminus \hat{M}_i} p(i,j) \leq \alpha \cdot \OPT(\instance).$
\end{proof}

By combining Lemma~\ref{lem:boundcap} and Lemma~\ref{lem:boundintersect}, we get that the mechanism is $2\alpha$-consistent.
  
\begin{lemma}
\label{lem:wconsistency}
\simplemech \ is a mechanism that is $2\alpha$-consistent.
\end{lemma}
\begin{proof}
Assume that $\instance = \hat{\instance}$. For any machine $i$, we have
$$\sum_{j \in M_i}p(i,j) =  \sum_{j \in M_i \cap \hat{M}_i}p(i,j) + \sum_{j \in M_i\setminus \hat{M}_i}p(i,j)\leq \alpha \OPT(\instance) + \alpha \OPT(\instance) = 2 \alpha \OPT(\instance)$$
where the inequality is by Lemma~\ref{lem:boundcap} and Lemma~\ref{lem:boundintersect}.
\end{proof}

\paragraph{The robustness.} The crucial part of the mechanism that allows obtaining a robustness of $n^2$ is that the scalars $r(i,j)$ are bounded between $1$ and $n$.

\begin{restatable}{rLem}{lemrobustnesslowerbound}
\label{lem:wrobustness}
\simplemech\ is a mechanism that is $\Theta(n^2)$-robust.
\end{restatable}

\begin{proof}
We first show that the robustness is at most $n^2$. Consider an instance $\instance$, an optimal assignment $\schedule^{\star}$ for $\instance$,  a job $j$, and a machine $i$ such that $j \in M_i$. Let $\optm_j$ be the machine that $j$ is assigned to according to $\schedule^{\star}$. By the mechanism, we have $1 \leq r(i',j) \leq n$ for all machines $i'$. We therefore have:
\[p(i,j) \leq r(i,j) \cdot p(i,j) \leq r(\optm_j,j) \cdot p(\optm_j,j) \leq n \cdot p(\optm_j,j),\]
where  the second inequality is by line~\ref{mech1:tie} and the fact that $j$ is assigned to $i$. We get that 
\[\ms(\instance, \schedule) = \max_{i \in N}\sum_{j \in M_i}p(i,j) \leq \sum_{i \in N}\sum_{j \in M_i}p(i,j)    \leq n \sum_{i \in N}\sum_{j \in M_i} p(\optm_j,j) \leq n^2\max_{i \in N}\sum_{j \in M_i} p(\optm_j,j) = n^2 \OPT.\]
In fact, the robustness guarantee is tight. Due to space limitation, we defer the proof to the appendix~\ref{sec:appwarmuplemma}. 
\end{proof}

\paragraph{Strategyproofness.} We show that the mechanism is strategyproof by proving that it is monotone and then using Lemma~\ref{lem:monotonicity}.

\begin{lemma}
\label{lem:wstrategyproof}
\simplemech \ is a strategyproof mechanism.
\end{lemma}
\begin{proof}
Consider a machine $i \in N$ and a predicted assignment $\hat{\schedule}$. Consider two instances $\instance$ and $\instance'$ that differ only on machine $i$ and the associated allocations $\schedule$ and $\schedule'$ returned by \simplemech \ when the predicted assignment is $\hat{\schedule}$. Given a fixed predicted assignment $\hat{\schedule}$, the mechanism is deterministic, so $x(i,j) \in \{0,1\}$ for all $i \in N$ and $j \in M$.  Since the predicted assignment is fixed, we also have that the scalars $r(i,j)$ are fixed.

Let  $j \in M$ be an arbitrary job.  Without loss of generality, we assume that $p(i,j) \leq  p'(i,j)$. Consider the case where $x(i,j) = 0$, i.e., job $j$ is not assigned to machine $i$ for processing times $\instance$. By the mechanism, this implies that $i \neq \argmin_{i'}  r(i',j) \cdot  p(i',j).$ Since (1) $p(i',j) = p'(i',j)$ for $i' \neq i$, (2) $p(i,j) \leq  p'(i,j)$, and (3)  $i \neq \argmin_{i'}  r(i',j) \cdot  p(i',j)$, we get that $i \neq \argmin_{i'}  r(i',j) \cdot  p'(i',j),$ which implies $x'(i,j) = 0$. We thus obtain that $x(i,j) \geq x'(i,j)$, which implies that $(x(i,j) - x'(i,j))(p(i,j) - p'(i,j)) \leq 0$ since $p(i,j) \leq  p'(i,j)$. \simplemech \ is thus a monotone mechanism, which, by Lemma~\ref{lem:monotonicity}, implies that it is strategyproof. 
\end{proof}

\paragraph{The main result for \simplemech.} Combining Lemmas~\ref{lem:wconsistency}, \ref{lem:wrobustness}, and \ref{lem:wstrategyproof}, we obtain the main result for \simplemech.

\begin{theorem}
    \simplemech \ is a strategyproof, $2\alpha$-consistent, and $\Theta(n^2)$-robust mechanism.
\end{theorem}

To obtain a polynomial-time mechanism, we can have $\alpha = 2$ and a consistency of $4$ with \algo \ being the $2$-approximation algorithm for makespan minimization on unrelated machines by \citet{LST90}. By ignoring running time considerations, we can obtain $\alpha = 1$ and a consistency of $2$ with \algo \ finding an optimal schedule.

\section{The Scaled Greedy Mechanism}\label{sec:scaledGreedy}
In this section, we present a deterministic polynomial-time mechanism that is $(4+2\tradeoff)$-consistent and $\left(1 + \frac{1}{\tradeoff}\right)n$-robust, where $\tradeoff \in \left(0, \frac{n}{2}-1\right)$ is a parameter that controls the tradeoff between the consistency and robusntess achieved by the mechanism. This mechanism  builds on the mechanism from the previous section. In the next section, we extend this mechanism to obtain a mechanism that achieves an approximation as a function of the prediction error.

\subsection{The mechanism} 

\mech, which is described formally in Mechanism~\ref{mech:scaledgreedy}, defines scalars $r(i,j)$  and assigns  a job $j$ to the machine $i$ with minimum scaled processing time $r(i,j) \cdot p(i,j)$. As in \simplemech, the mechanism first computes a predicted assignment $\hat{\schedule}$ for the predicted instance $\hat{\instance}$. Recall that  scalars $r(i,j) = \frac{\hat{p}(\poptm_j,j)}{\hat{p}(i,j)}$ are designed so that the mechanism assigns a job $j$ to its predicted assignment $\poptm_j$ when the predictions are correct,  even if $j$ has a smaller predicted processing time on $i$ than on $\poptm_j$.
To improve the robustness from quadratic to linear, \mech \ upper bounds the scalars $r(i,j)$ more aggressively, and more carefully, than \simplemech, which simply upper bounds them by $n$. Scalars that are close to $1$ are intuitively better for the robustness of the mechanism since the greedy mechanism, which achieves an $n$-approximation, corresponds to the mechanism where the scalars are all equal to $1$. The mechanism upper bounds the scalars by maintaining sets $J_i$, $I$ and $T$. 

\vspace{.3cm}
\begin{algorithm}[H]
\textbf{Input:} instance $\instance \in \mathbb{R}^{n \times m}$, predicted instance $\hat{\instance} \in \mathbb{R}^{n \times m}$, scheduling algorithm \algo, prediction confidence parameter $\gamma \in (0, \frac{n}{2} - 1)$

$\hat{\schedule} \leftarrow \algo(\hat{\instance})$\\
$\poptm_j \leftarrow  $ the machine $i$ that job $j$ is assigned  to according to $\hat{\schedule}$, i.e., $\hat{x}(\poptm_j,j) = 1$, for each $j \in M$ \\
$r(i,j) \leftarrow \max\left( \frac{\hat{p}(\poptm_j,j)}{\hat{p}(i,j)},1\right)$ for each $(i,j) \in N \times M$ \label{line:update1}\\
$J_{i} \leftarrow \emptyset$ for all $i \in N$ \\
$I \leftarrow N$ \\
$T \leftarrow \{(i,j): j  \in M,  i = \argmin_{i' \in I : \hat{p}(i',j) < \hat{p}(\poptm_j, j)}\hat{p}(i',j)\}$\\
\While{$T$ is not empty}{
    $(i^*, j^*) \leftarrow \argmax_{(i,j) \in T} \frac{\hat{p}(\poptm_j,j)}{\hat{p}(i,j)}$ \\
        $J_{i^*} \leftarrow J_{i^*} \cup \{j^*\}$ \\
        \For{all $i$ s.t. $\hat{p}(i^*, j^*) \leq \hat{p}(i,j^*)$ \label{line:forloop}}{
             update $r(i, j^*) \leftarrow 1$ \label{line:update2}
        }
        $I \leftarrow \{i \in N : \sum_{j \in J_{i}}\hat{p}(i,j) < \tradeoff \ms(\hat{\instance},\hat{\schedule})\}$\label{line:Idef} \\
    $T \leftarrow \{(i,j): j  \in M \setminus (\cup_i J_i),  i = \argmin_{i' \in I : \hat{p}(i',j) < \hat{p}(\poptm_j, j)}\hat{p}(i',j)\}$ \label{line:Tcondition}\\
    }
$i_j \leftarrow \argmin_i  r(i,j) \cdot  p(i,j) $, break ties in favor of  $\poptm_j$ first and $i'$ if $j \in J_{i'}$ second, for each $j \in M$\label{line:tiebreaking}\\
\textbf{if} $i = i_j$ \textbf{then} $x(i, j) = 1$ \textbf{else} $x(i,j) = 0$, for each $(i,j) \in N \times M$ \\
\Return{$\mathbf{x}$}
\caption{\mech}
\label{mech:scaledgreedy}
\end{algorithm}
\vspace{.3cm}
$J_i$ is the set of jobs $j$ for which \mech \ assigns $j$ to $i$ when the predictions are correct even though their predicted assignment is not $i$. When \mech \ adds a job $j^{*}$ to $J_{i^*}$, it allows to set the scalars $r(i,j^*)$ to $r(i,j^*) = 1$, for all machines $i$ such that $\hat{p}(i^*, j^*) \leq \hat{p}(i,j^*)$, while guaranteeing that job $j^*$ is assigned to $i^*$ when the predictions are correct. The sets $J_i$ are initially empty. To maintain a good consistency, we bound the total predicted processing time of the jobs in a set $J_i$ by  $\tradeoff \ms(\hat{\instance},\hat{\schedule})$. $I$ is the set of machines for which this bound has not yet been violated, i.e., the machines $i$ that are such that we can still add jobs to $J_i$. $I$ initially consists of all the machines.  $T$ is the set of pairs $(i,j)$ such that job $j$ is a candidate to be added to $J_i$. A pair $(i,j)$ is in $T$ if (1) $j$ has not yet been assigned by \mech \ to another machine, i.e., $j \not \in \cup_i J_i$, (2) the upper bound on the total predicted processing time of the jobs already in $J_i$ has not yet been violated, i.e., $i \in I$, (3) $j$ has a smaller processing time on $i$ than on its predicted assignment $\poptm_j$, and (4) $i$ is the machine with minimum predicted processing time $\hat{p}(i,j)$ among the machines in $I$. At each iteration, job $j^{*}$ is added $J_{i^{*}}$ where $(i^{*}, j^{*})$ is the pair $(i,j)$ with maximum $\frac{\hat{p}(\poptm_j,j)}{\hat{p}(i,j)}$ in $T$.  For each job $j$, if there are multiple machines with minimum scaled processing time  $r(i,j) \cdot  p(i,j)$, ties are broken in favor of, in order of priority, (1) machine $\poptm_j$ and, if there is one, (2) the machine $i'$ such that $j \in J_{i'}$, and (3) the machine with minimum index.

\subsection{The analysis} 

We first analyze the mechanism's consistency and robustness, and finally show it is strategyproof. The following simple observation is used for both the consistency and the robustness analyses.
 
\begin{observation}\label{obs:rbiggerthan1}
For every machine $i \in N$ and job $j \in M$, at every iteration of \mech, we have $r(i,j) \geq  1$ and, if $\hat{p}(i,j) \geq \hat{p}(\poptm_j, j), r(i,j) =  1$. 
\end{observation}

\paragraph{The consistency.} As in the analysis of \simplemech, we assume that $\algo$   is an $\alpha$-approximation algorithm for makespan minimization and let $M_i = \{j \in M: x(i,j)  = 1\}$ and $\hat{M}_i = \{j \in M:  \hat{x}(i, j) = 1\}$ be the sets of jobs $j$  such that $j$ is assigned to machine $i$ by \mech \ and the predicted assignment, respectively.  To bound the total processing time of the jobs $M_i$ assigned to a machine $i$,   we again separately bound the total processing time of  $M_i \cap \hat{M}_i$ and $M_i \setminus \hat{M}_i$. The next lemma,  which bounds the total processing time of  $M_i \cap \hat{M}_i$ when the predictions are correct, is identical as Lemma~\ref{lem:boundcap} for \simplemech, and its proof is also identical to the proof of Lemma~\ref{lem:boundcap}.

\begin{lemma}
\label{lem:boundcap2}
Assume that $\hat{\instance} = \instance$, then, for all $i \in N$, $\sum_{j \in M_i \cap \hat{M}_i} p(i,j) \leq \alpha \OPT(\instance)$.
\end{lemma}

Next, to bound the total processing time of $M_i \setminus \hat{M}_i$, we first characterize the jobs in $M_i \setminus \hat{M}_i$ as being exactly the jobs in $J_i$.

\begin{lemma}\label{lem:jobassignment}
Assume that $\hat{\instance} = \instance$, then, for every machine $i \in N$, we have that $J_i = M_i \setminus \hat{M}_i$.
\end{lemma}
\begin{proof}
Consider a machine $i \in N$ and a job $j \in J_i$. We want to show that $j \in M_i \setminus \hat{M}_i$. Note that since $j \in J_i$, we had $(i,j) \in T$ at some iteration of the mechanism, which implies by Line~\ref{line:Tcondition} that $\hat{p}(i,j) < \hat{p}(\poptm_j, j)$. Since $\hat{p}(i,j) < \hat{p}(\poptm_j, j)$, we have that $i \neq \poptm_j$. Since $\hat{x}(\poptm_j, j) = 1$ by definition of $\poptm_j$, we have that $\hat{x}(i,j) = 0$, which implies $j \not \in \hat{M}_i$. 

Next, we show that $j \in M_i$, which we show by proving that $r(i,j) p(i,j) \leq r(i', j) p(i',j)$ for all $i' \not \in \{i, \poptm_j\}$  and $r(i,j) p(i,j) < r(\poptm_j, j) p(\poptm_j,j)$. There are two cases of machines $i' \not \in  \{i, \poptm_j\}$.  The first case is if $i'$ is such that $\hat{p}(i',j) < \hat{p}(i,j)$, which is the main part of the proof. Since $\hat{p}(i,j) < \hat{p}(\poptm_j, j)$ (which was shown earlier in this proof), we have $\hat{p}(i',j) < \hat{p}(\poptm_j, j)$, so $r(i',j)$ was initialized to $r(i',j) = \frac{\hat{p}(\poptm_j,j)}{\hat{p}(i',j)}$ in line~\ref{line:update1}. For all $i'' \neq i$, we have that $j \not \in J_{i''}$ since, by line~\ref{line:Tcondition}, $(i'', j) \not \in T$ if $j \in J_i$, which implies that $j$ cannot be added to $J_{i''}$. Thus, the only iteration of the mechanism where $j^{\star} = j$ is when $i^{\star} = i$, and since $\hat{p}(i',j) < \hat{p}(i,j)$, this implies that $r(i', j)$ is never updated in line~\ref{line:update2}. Together, with the fact that it was initialized to $r(i',j) = \frac{\hat{p}(\poptm_j,j)}{\hat{p}(i',j)}$, we have that $r(i',j) = \frac{\hat{p}(\poptm_j,j)}{\hat{p}(i',j)}$ when the mechanism terminated. Thus, when $\instance = \hat{\instance}$, we have
\begin{align}\label{eq:si1forsure}
    r(i,j) p(i,j)   = 1\cdot \hat{p}(i',j)  < \hat{p}(\poptm_j, j) =\frac{\hat{p}(\poptm_j, j)}{\hat{p}(i',j)} \cdot \hat{p}(i',j) = r(i',j) p(i',j).
\end{align}
where the first equality is since $r(i,j)=1$ (because $j \in J_i$). The second case of $i' \not \in \{i, \poptm_j\}$ is if $i'$ such that $\hat{p}(i',j) \geq \hat{p}(i,j)$. In this case, we have that
\begin{align}\label{eq:si1forsure1}
  p(i,j) \cdot r(i,j) = \hat{p}(i,j) \leq \hat{p}(i',j) \leq p(i',j) \cdot r(i',j),  
\end{align}
where the last inequality is by Observation~\ref{obs:rbiggerthan1}. Next, note that
\begin{align}\label{eq:si1forsure2}
p(i,j) \cdot r(i,j) = \hat{p}(i,j) < \hat{p}(\poptm_j,j) = p(\poptm_j,j) \cdot r(\poptm_j, j), 
\end{align}
where the inequality is since, at some iteration of the mechanism, $(i,j) \in T$ and the last equality since  $r(\poptm_j, j) = 1.$
Since we have shown that $r(i,j) p(i,j) \leq r(i', j) p(i',j)$ for all $i' \not \in \{i, \poptm_j\}$  and $r(i,j) p(i,j) < r(\poptm_j, j) p(\poptm_j,j)$, we have that, by line~\ref{line:tiebreaking} of the mechanism, $x(i,j) = 1$, which implies that $j \in M_i$. Since we have shown that $j \not \in \hat{M}_i$ and $j \in M_i$, we get $j \in M_i \setminus \hat{M}_i$.

It remains to show that if $j \in  M_i \setminus \hat{M}_i$, then $j \in J_i$. Consider an arbitrary job $j \in \ M_i \setminus \hat{M}_i$. Since $j \not \in \hat{M}_i$, $i \neq \poptm_j.$ Since $j \in M_i$ and $i \neq \poptm_j$, we have that $$r(i,j) p(i,j) < r(\poptm_j, j) p(\poptm_j, j) = p(\poptm_j,j)$$ where the inequality is by line~\ref{line:tiebreaking} of the mechanism and the equality since  $r(\poptm_j, j) = 1$. Thus, $r(i,j) < \frac{p(\poptm_j,j)}{ p(i,j)}$, which implies $r(i,j) = 1$. Thus there is some iteration of the mechanism where $j^{\star} = j$ got added to $J_{i^{\star}}$ for some machine $i^{\star}$. Assume by contradiction that $i^* \neq i$, so $j \in J_{i^{\star}}$ for $i^* \neq i$. Then, by the first part of this proof, we have that  $j \in M_{i^{\star}}$, which is a contradiction with  $j \in M_i$. Thus, $i^* = i$ and $j \in J_i$.
\end{proof}

The next lemma bounds the total processing time of jobs in $M_i \setminus \hat{M}_i$ by using the fact that $J_i = M_i \setminus \hat{M}_i$ and by bounding the total processing time of jobs in $J_i$.

\begin{lemma}\label{lem:upperboundvalforsi}
For any input parameter $\tradeoff \geq 0$, then, for every machine $i \in N$, if $J_i = M_i \setminus \hat{M_i}$, we have that $\sum_{j \in M_i \setminus \hat{M}_i} \hat{p}(i,j) < (1+\tradeoff)\alpha \OPT(\hat{\instance}).$
\end{lemma}
\begin{proof}
Consider an arbitrary machine $i$ and
let $j'$ denote the last job added into $J_i$ by the mechanism. First observer that
$$   \sum_{j \in M_i \setminus \hat{M}_i} \hat{p}(i,j) = \sum_{j \in J_i} \hat{p}(i,j) = \sum_{j \in J_i \setminus \{j'\}} \hat{p}(i,j) + \hat{p}(i,j') < \tradeoff\ms(\hat{\instance},\hat{\schedule}) + \hat{p}(i,j')  < \tradeoff\ms(\hat{\instance},\hat{\schedule}) + \hat{p}(\poptm_{j'},j')  $$
where the first inequality is by Line~\ref{line:Idef} and Line~\ref{line:Tcondition} of the mechanism, and the last inequality is  since $j' \in J_i$. Next, we have
\begin{align*}
 \tradeoff\ms(\hat{\instance},\hat{\schedule}) + \hat{p}(\poptm_{j'},j')    \leq (1+\tradeoff)\ms(\hat{\instance},\hat{\schedule}) \leq (1+\tradeoff)\alpha \OPT(\hat{\instance})
\end{align*}
where the first inequality is since the makespan of an assignment is at least the processing time of a job $j$ on a machine $i$, for any $j$ assigned to $i$, the second inequality since \algo \ is an $\alpha$-approximation algorithm.
\end{proof}

By combining Lemma~\ref{lem:boundcap2} to bound $\sum_{j \in M_i \cap \hat{M_i}} p(i,j)$ and  Lemma~\ref{lem:upperboundvalforsi} to bound $\sum_{j \in M_i \setminus \hat{M_i}} p(i,j)$, the next lemma bounds $\sum_{j \in M_i} p(i,j)$  by $(2+\tradeoff)\alpha$ for any machine $i$, which implies that \mech \ is  $(2+\tradeoff)\alpha$-consistent.

\begin{lemma}\label{lem:consistency}
For any input parameter $\tradeoff \geq 0$,  the \mech\ mechanism is $(2+\tradeoff)\alpha$-consistent.
\end{lemma}
\begin{proof}
Assume that $\instance = \hat{\instance}$ and consider an arbitrary machine $i$, then we have $J_i = M_i \setminus \hat{M}_i$ by Lemma~\ref{lem:jobassignment}. Thus
\[ \sum_{j \in M_i} p(i,j) = \sum_{j \in M_i \cap \hat{M}_i} p(i,j)  +\sum_{j \in M_i\setminus\hat{M}_i} p(i,j) \leq  \alpha \OPT(\instance) +(1+\tradeoff)\alpha \OPT(\instance)  = (2+\tradeoff)\alpha\OPT,\]
where the inequality is by Lemma~\ref{lem:boundcap2}, Lemma~\ref{lem:upperboundvalforsi} and the assumption that $\instance = \hat{\instance}$.
\end{proof}

\paragraph{The robustness.} Recall that the main reason that \simplemech \ achieved bounded robustness is thanks to the scalars being individually upper bounded by $n$. To show that \mech \ achieves linear robustness, we give a stronger upper bound on the sum of the largest scalar on each machine. More precisely, in Lemma~\ref{lem:sumscalar}, we show that $\sum_{i : \max_{j}r(i,j) > 1} \max_{j}r(i,j) \leq \frac{n}{\tradeoff}$, which is the main lemma to show the linear robustness. To prove Lemma~\ref{lem:sumscalar}, we first need two helper lemmas. The first is that for a fixed machine $i$, the jobs added to $J_i$, i.e. the jobs $j$ with scalars $r(i,j)$ modified from $\frac{\hat{p}(\poptm_j,j)}{\hat{p}(i,j)}$ to $1$ are the jobs with largest ratio $\frac{\hat{p}(\poptm_j,j)}{\hat{p}(i,j)}$.

\begin{lemma}\label{lem:decreasingratio}
For any machine $i \in N$, job $j \in J_i$, and job $j' \not \in J_i$, if $r(i, j') > 1$, then
\[\frac{\hat{p}(\poptm_j,j)}{\hat{p}(i,j)} \geq \frac{\hat{p}(\poptm_{j'},j')}{\hat{p}(i,j')}.\]
\end{lemma}

\begin{proof}
Assume for contradiction that $\frac{\hat{p}(\poptm_{j'},{j'})}{\hat{p}(i,{j'})} > \frac{\hat{p}(\poptm_j,j)}{\hat{p}(i,j)}$ for some $j \in J_i$ and $j' \not \in J_i$ such that  $r(i, j') > 1$. Since $j \in J_i$, we have that $(i,j) \in T$ at some iteration of the mechanism. We let the iteration when $j$ is added to $J_i$ by the mechanism be $t$. First note that $\hat{p}(i,{j'}) < \hat{p}(\poptm_{j'},{j'})$, otherwise $\frac{\hat{p}(\poptm_{j'},{j'})}{\hat{p}(i,{j'})} \leq 1 < \frac{\hat{p}(\poptm_j,j)}{\hat{p}(i,j)}$ (where the strict inequality is since $j \in J_i$), which is a contradiction. 

 Since $j$ was added to $J_i$ at iteration $t$, we have $i \in I$ at iteration $t$, which implies $i \in I$ at all iterations $t' \leq t$. There are two cases. If, at iteration $t$, $j' \in J_{i'}$ for some machine $i'$, then $j'$ was added to  $J_{i'}$ at some iteration $t' < t$. Thus, $(i', j') \in T$ at iteration $t'$. By line~\ref{line:Tcondition}, this implies that $\hat{p}(i', j') \leq \hat{p}(i,j')$ since $i \in I$ at iteration $t'$. By line~\ref{line:update2}, $j'$ added to  $J_{i'}$ and $\hat{p}(i', j') \leq \hat{p}(i,j')$ imply $r(i, j') = 1$, which is a contradiction. 
Thus, at iteration $t$, $j' \not \in \cup_{i'} J_{i'}$. 

Since $\hat{p}(i,{j'}) < \hat{p}(\poptm_{j'},{j'})$ and $j' \not \in \cup_{i'} J_{i'}$,  there exists machine $i'$ such that $(i', j') \in T$ at iteration $t$ and such that $\hat{p}(\poptm_{j'},{j'}) \leq \hat{p}(i,{j'})$. We get that
$$\frac{\hat{p}(\poptm_{j'}, j')}{\hat{p}(i', j')} \geq \frac{\hat{p}(\poptm_{j'}, j')}{\hat{p}(i, j')} > \frac{\hat{p}(\poptm_j,j)}{\hat{p}(i,j)}$$
where the inequality is by the assumption of the proof.
Since $(i', j') \in T$  and $\frac{\hat{p}(\poptm_{j'}, j')}{\hat{p}(i', j')} > \frac{\hat{p}(\poptm_j,j)}{\hat{p}(i,j)},$ we have that $(i^*, j^*) \neq (i,j)$ at iteration $t$, which is a contradiction with $j$ that is added to $J_i$ at iteration $t$.
\end{proof}

The second helper lemma shows that if there is a job $j$ with scalar $r(i,j) > 1$ with respect to machine $i$, i.e. a scalar that the mechanism did not decrease to $1$, then we have that $\sum_{j' \in J_i}\hat{p}(i,j') \geq \tradeoff \ms(\hat{\instance}, \hat{\schedule})$, i.e., the upper bound $\tradeoff \ms(\hat{\instance}, \hat{\schedule})$ on the total processing time of jobs in $J_i$ has been violated. Informally, the scalar $r(i,j)$ couldn't be decreased to $1$ because $J_i$ was already full.

\begin{lemma}\label{lem:setasideopt} For any $\tradeoff > 0$, when \mech \ terminates, for any machine $i$, if there is a job $j$ such that $r(i,j)>1$,  then $\sum_{j' \in J_i}\hat{p}(i,j') \geq \tradeoff \ms(\hat{\instance}, \hat{\schedule})$.
\end{lemma}

\begin{proof}
Assume for contradiction, that there exist a machine $i$ with $r(i,j')>1$ for some job $j'$ and $\sum_{j \in J_i}\hat{p}(i,j)< \tradeoff\ms(\hat{\instance}, \hat{\schedule})$. In other words, $i \in I$ throughout the execution of the mechanism. Since $r_{i,j'}>1$, we have $\hat{p}(i,j') < \hat{p}(\poptm_{j'},{j'})$, otherwise $r_{i,j'}=1$ by line~\ref{line:update1}. If, at the end of the execution of \mech, $j' \in J_{i'}$ for some machine $i'$, then $j'$ was added to  $J_{i'}$ at some iteration $t'$. Thus, $(i', j') \in T$ at iteration $t'$. By line~\ref{line:Tcondition}, this implies that $\hat{p}(i', j') \leq \hat{p}(i,j')$ since $i \in I$ at iteration $t'$. By line~\ref{line:update2}, $j'$ added to  $J_{i'}$ and $\hat{p}(i', j') \leq \hat{p}(i,j')$ imply $r(i, j') = 1$, which is a contradiction. 
Thus, when \mech\ terminates, $j' \not \in \cup_{i'} J_{i'}$. 

Since $\hat{p}(i,{j'}) < \hat{p}(\poptm_{j'},{j'})$ and $j' \not \in \cup_{i'} J_{i'}$,  either there exists machine $i'$ such that $(i', j') \in T$ at the end of the execution and such that $\hat{p}(\poptm_{j'},{j'}) \leq \hat{p}(i,{j'})$ or $(i,j') \in T$. However, both cases can't happen since the mechanism only terminates when $T$ is empty.
\end{proof}

The next lemma is the main lemma for the robustness analysis, it upper bounds the sum of the largest scalars on each machine. To prove it, we use the two previous helper lemmas.

\begin{lemma}\label{lem:sumscalar} For any $\tradeoff > 0$, when \mech \ terminates, we have that
\[\sum_{i  : \max_j r(i,j) > 1} \max_j r(i,j) \leq \frac{n}{\tradeoff}.\]
\end{lemma}
\begin{proof}
First note that the sets $J_i$ for $i \in N$ are disjoint, since once any $j$ is added to $J_i$ for some machine $i$, $j$ is no longer consider for $T$ by line~\ref{line:Tcondition}. We therefore have:
\begin{align*}
    \ms(\hat{\instance},\hat{\schedule}) &= \max_i \sum_{j: \poptm_j = i}\hat{p}(\poptm_j,j) \geq \frac{\sum_{j}\hat{p}(\poptm_j,j)}{n} \geq \frac{\sum_{i  : \max_j r(i,j) > 1} \sum_{j \in J_i}\hat{p}(\poptm_j,j)}{n},\numberthis \label{eq:sumofopt1}
\end{align*}
Where the first equation is by the definition of predicted machine, the first inequality is because the maximum is weakly greater than the average, and the last inequality is due to the fact that $J_i$ and $J_{i'}$ are disjoint for any two distinct machines $i,i'\in N$. By Lemma~\ref{lem:decreasingratio}, for any $j \in J_i$, we have 
\begin{align*}
    \hat{p}(\poptm_j,j) = \hat{p}(i,j) \cdot  \frac{\hat{p}(\poptm_j,j)}{\hat{p}(i,j)} \geq \hat{p}(i,j) \cdot \max_j r(i,j). \numberthis \label{eq:decreasingratio}
\end{align*}
Combining Inequalities \eqref{eq:sumofopt1} and \eqref{eq:decreasingratio} we get
\[\ms(\hat{\instance},\hat{\schedule}) \geq \frac{\sum_{i  : \max_j r(i,j) > 1} \max_j r(i,j) \sum_{j \in J_i}\hat{p}(i,j)}{n}.\]
By Lemma~\ref{lem:setasideopt} we know that for any machine $i$ such that $\max_j r(i,j) > 1$, $\sum_{j \in J_i}\hat{p}(i,j) \geq \tradeoff\ms(\hat{\instance},\hat{\schedule})$. Substitute it in the equation above we get that:
\[\ms(\hat{\instance},\hat{\schedule}) \geq \frac{\sum_{i  : \max_j r(i,j) > 1} \max_j r(i,j) \cdot \tradeoff \ms(\hat{\instance},\hat{\schedule})}{n} \quad \Rightarrow \quad \sum_{i  : \max_j r(i,j) > 1} \max_j r(i,j) \leq \frac{n}{\tradeoff}.\qedhere\]
\end{proof}
We are now ready to show the robustness of the algorithm. 
\begin{lemma}\label{lem:robustness}
For any $\tradeoff > 0$,  \mech~is a mechanism that is $(1+\frac{1}{\tradeoff})n$-robust.
\end{lemma}
\begin{proof}
First note that since $\OPT$ is the maximum finishing time of any machine in the optimal schedule, and since the maximum is greater than the average we have:
\begin{align}\label{eq:optlowerbound}
    \OPT \geq \frac{\sum_{j \in M}p(\optm_j, j)}{n}.
\end{align}
We now partition the set of jobs $M$ into two subsets, let $S^*$ be the set of jobs $j$ such that $r(\optm_j, j)=1$ and $M \setminus S^*$ be the rest of the jobs $j$ with $r(\optm_j,j) >1$.

For any job $j \in S^*$, let $i_j$ be the machine the mechanism assigns $j$ on, i.e., $x(i_j,j) = 1$. we have
\begin{align}\label{eq:m1processingtime}
    p(i_j, j ) \leq r(i,j) p(i_j, j) \leq r(\optm_j, j) p(\optm_j, j) = p(\optm_j, j),
\end{align}
where the first inequality is by Observation~\ref{obs:rbiggerthan1}, the second inequality is by the fact that $x(i_j,j)=1$ and the last equality is by definition of $S^*$.
Combining Inequalities \eqref{eq:optlowerbound} and \eqref{eq:m1processingtime} we get
\begin{align}\label{eq:m1sum}
    \sum_{j \in S^*} p(i_j, j)  \leq   \sum_{j \in S^*}p(\optm_j, j)  \leq  \sum_{j \in M}p(\optm_j, j) \leq n \cdot \OPT.
\end{align}

Now consider any job $j \in M \setminus S^*$, first note that since $r(\optm_j, j)>1$, $\optm_j \in \{i  : \max_j r(i,j) > 1\}$. Similarly as the analysis of $S^*$ we have 
\begin{align}\label{eq:m2processingtime}
    p(i_j, j ) \leq r(i,j) p(i_j, j) \leq r(\optm_j, j) p(\optm_j, j)
\end{align}
 Now consider the optimal makespan $\OPT$ again, by the definition of optimal schedule, we have:
\begin{align*}
    \OPT = \max_i \sum_{j: \optm_j = i} p(\optm_j,j) \geq \max_{i  : \max_j r(i,j) > 1} \sum_{j: \optm_j = i} p(\optm_j,j). \numberthis \label{eq:optlowerboundm2}
\end{align*}
Combining Inequalities \eqref{eq:m2processingtime} and \eqref{eq:optlowerboundm2}, we get:
\begin{align*}
    \sum_{j \in M \setminus S^*} p(i_j,j) & \leq \sum_{j \in M \setminus S^*} r(\optm_j, j)p(\optm_j, j) \tag{by Inequality \eqref{eq:m2processingtime}}\\
    & \leq \sum_{i  : \max_j r(i,j) > 1} \sum_{j: \optm_j = i} r(\optm_j, j)p(\optm_j, j) \\
    & \leq \sum_{i  : \max_j r(i,j) > 1} (\max_j r(i,j)) \cdot \sum_{j: \optm_j = i} p(\optm_j, j) \\
    & \leq \sum_{i  : \max_j r(i,j) > 1}(\max_j r(i,j))\OPT \tag{by Inequality \eqref{eq:optlowerboundm2}}\\
    & \leq \frac{n}{\tradeoff} \cdot \OPT, \numberthis \label{eq:m2sum} 
\end{align*}
where the last inequality is due to Lemma~\ref{lem:sumscalar}.
Combining Inequalities \eqref{eq:m1sum} and \eqref{eq:m2sum} we get:
$$\max_i \sum_{j \in M_i} p(i, j) \leq \sum_{j \in M} p(i_j, j) \leq \sum_{j \in S^*} p(i_j, j) + \sum_{j \in M \setminus S^*} p(i_j, j) \leq \left(1+\frac{1}{\tradeoff}\right)n \cdot \OPT.  \eqno \qedhere$$
\end{proof}

\paragraph{Strategyproofness.}
\begin{lemma}
\label{lem:astrategyproof}
\mech \ is a strategyproof mechanism.
\end{lemma}
\begin{proof}
Consider a machine $i \in N$ and a predicted assignment $\hat{\schedule}$. Consider two instances $\instance$ and $\instance'$ that differ only on machine $i$ and the associated allocations $\schedule$ and $\schedule'$ returned by \mech \ when the predicted assignment is $\hat{\schedule}$. Given a fixed predicted assignment $\hat{\schedule}$, the mechanism is deterministic, so $x(i,j) \in \{0,1\}$ for all $i \in N$ and $j \in M$.  Note that line 4 to 14 don't use $\instance$, thus the scalars are, as for \simplemech, independent of  $\instance$. The remaining of the proof follows identically as the analysis of \simplemech \ being strategyproof (Lemma~\ref{lem:wstrategyproof})
\end{proof}

\paragraph{The main result for \mech.}
Combining Lemmas~\ref{lem:consistency}, \ref{lem:robustness}, and \ref{lem:astrategyproof}, we obtain the main result for \mech, which is that it is a strategyproof mechanism that, with input parameter $\tradeoff$ being any constant, achieves constant consistency and linear robustness.

\begin{theorem}
    For any $\tradeoff \in \left(0, \frac{n}{2}-1\right)$, \mech \ is a strategyproof, $(2 + \tradeoff) \alpha$-consistent, and $\left(1+\frac{1}{\tradeoff}\right)n$-robust algorithm.
\end{theorem}

To obtain a polynomial-time algorithm, we can have $\alpha = 2$ and a consistency of $4 + 2 \tradeoff$ with \algo \ being the $2$-approximation algorithm for makespan minimization on unrelated machines by \citet{LST90}. By ignoring running time consideration, we can obtain $\alpha = 1$ and a consistency of $2 + \tradeoff$ with \algo \ finding an optimal schedule.

\section{The Error Tolerant Scaled Greedy Mechanism}\label{sec:error}
In this section, we give a mechanism that builds on the mechanism from the previous section and achieves a constant approximation not only when the predictions are accurate, but also when the predictions are approximately accurate. Recall from Section~\ref{sec:prelims} that, given a prediction $\hat{\instance}$ of the actual instance $\instance$,  the prediction error $\eta$ is the largest ratio between the predicted processing time and actual processing time for any $i, j$ pair, i.e.,
$\eta = \max_{i \in M, j \in N} \max \left\{\frac{\hat{p}_{ij}}{p_{ij}}, \frac{p_{ij}}{\hat{p}_{ij}}\right\}.$ This mechanism takes as input an error tolerance parameter $\errbound > 0$ and achieves an approximation of  $(2+\tradeoff)\alpha\eta^2$ if $\eta \leq \errbound$ and $(1+\frac{1}{\tradeoff})\errbound^2n$ otherwise, where, similarly as in the previous section, $\tradeoff$ is the prediction confidence parameter and $\alpha$ is the approximation of the scheduling algorithm used in the mechanism.

\paragraph{The mechanism.} \mecherr, formally described in Mechanism~\ref{mech:errorversionscaledgreedy}, is similar to \mech, but takes as input an additional parameter $\errbound > 0$ that is the error tolerance of the mechanism. Recall from the previous section that if $j \in J_i$ for some machine $i$, then, when the predictions are correct, \mech \ assigns $j$ to machine $i$, otherwise it assigns $j$ to its predicted assignment $\poptm_j.$ \mecherr \ aims to achieve this assignment not only when the predictions are correct, but also when they are approximately correct, with one exception that is later discussed. In order to achieve this assignment when the predictions are approximately correct, the mechanism sets the scalars $r(i,j)$ of pairs $i,j$ such that job $j$ that should be assigned to $i$ to $r(i,j) = \frac{1}{\errbound^2}$ in lines~\ref{eline:update2} and \ref{eline:enforce2}. The exception previously mentioned is that it is not only the scalar of $j^{\star} \in J_{i^{\star}}$ that is updated to  $\frac{1}{\errbound^2}$ in line~\ref{eline:update2}, but the scalars $r(i, j^{\star})$ for all machines $i$ such that $\hat{p}(i^*, j^*) \leq \hat{p}(i,j^*)$. The remainder of the mechanism is identical to \mech.
\begin{algorithm}[h]
\textbf{Input:} instance $\instance \in \mathbb{R}^{n \times m}$, predicted instance $\hat{\instance} \in \mathbb{R}^{n \times m}$, scheduling algorithm \algo, prediction confidence parameter $\gamma \in (0, \frac{n}{2} - 1)$, error tolerance parameter $\errbound > 0$

$\hat{\schedule} \leftarrow \algo(\hat{\instance})$\\
$\poptm_j \leftarrow  $ the machine $i$ that job $j$ is assigned  to according to $\hat{\schedule}$, i.e., $\hat{x}(\poptm_j,j) = 1$, for each $j \in M$ \\
$r(i,j) \leftarrow \max\left( \frac{\hat{p}(\poptm_j,j)}{\hat{p}(i,j)},1\right)$ for each $(i,j) \in N \times M$ \label{eline:update1}\\
$J_{i} \leftarrow \emptyset$ for all $i \in N$ \\
$I \leftarrow N$ \\
$T \leftarrow \{(i,j): j  \in M,  i = \argmin_{i' \in I : \hat{p}(i',j) < \hat{p}(\poptm_j, j)}\hat{p}(i',j)\}$\\
\While{$T$ is not empty}{
    $(i^*, j^*) \leftarrow \argmax_{(i,j) \in T} \frac{\hat{p}(\poptm_j,j)}{\hat{p}(i,j)}$ \\
        $J_{i^*} \leftarrow J_{i^*} \cup \{j^*\}$ \\
        \For{all $i$ s.t. $\hat{p}(i^*, j^*) \leq \hat{p}(i,j^*)$}{
             update $r(i, j^*) \leftarrow \frac{1}{\errbound^2}$ \label{eline:update2}
        }
        $I \leftarrow \{i \in N : \sum_{j \in J_{i}}\hat{p}(i,j) < \tradeoff \ms(\hat{\instance},\hat{\schedule})\}$\label{eline:Idef} \\
    $T \leftarrow \{(i,j): j  \in M \setminus (\cup_i J_i),  i = \argmin_{i' \in I : \hat{p}(i',j) < \hat{p}(\poptm_j, j)}\hat{p}(i',j)\}$ \label{eline:Tcondition}\\
    }
\For{all jobs $j$ s.t $j \notin \cup_{i \in N} J_i$\label{eline:enforce1}}{
    update $r(\poptm_j,j)  \leftarrow  \frac{1}{\errbound^2}$\label{eline:enforce2}
}
$i_j \leftarrow \argmin_i  r(i,j) \cdot  p(i,j) $, break ties in favor of  $\poptm_j$ first and $i'$ if $j \in J_{i'}$ second, for each $j \in M$ \label{eline:tie} \\
\textbf{if} $i = i_j$ \textbf{then} $x(i, j) = 1$ \textbf{else} $x(i,j) = 0$, for each $(i,j) \in N \times M$\ \\
\Return{$\mathbf{x}$}
\caption{\mecherr}
\label{mech:errorversionscaledgreedy}
\end{algorithm}

\paragraph{The analysis.} Note that mechanisms \mech\ and \mecherr\ are identical except lines~\ref{eline:update2}, \ref{eline:enforce1} and \ref{eline:enforce2}. It is easy to verify that Lemma~\ref{lem:decreasingratio}, Lemma~\ref{lem:setasideopt}, and Lemma~\ref{lem:sumscalar} for \mech \ also hold for \mecherr. The Observation~\ref{obs:rbiggerthan1} also holds true for \mecherr\ for any $r(i,j)$ given that $i \neq \poptm_j$ and $j \notin J_i$. We first show that the characterization of the jobs in $J_i$ as  $J_i = M_i \setminus \hat{M}_i$ for \mech \ when the predictions are accurate now also holds for \mecherr \ when the error is bounded by the error tolerance, i.e., $\eta(\instance, \hat{\instance}) \leq \errbound$.

\begin{lemma}
\label{lem:errorJ}
Given any prediction $\hat{\instance}$ and actual instance $\instance$, if $\eta(\instance, \hat{\instance}) \leq \errbound$, then, for every machine $i \in N$, we have $J_i = M_i \setminus \hat{M}_i$.
\end{lemma}
\begin{proof}
 By Inequality \eqref{eq:si1forsure}, \eqref{eq:si1forsure1} and \eqref{eq:si1forsure2} from the proof of Lemma~\ref{lem:jobassignment} we know that if $j \in J_i$ for some machine $i'$, $\hat{p}(i',j) \leq \hat{p}(i,j) \cdot r(i,j)$ for any $i,j$ pair. By line~\ref{eline:update2} of the mechanism and the fact that the mechanism did not modify $r(i,j)$ for any other machine $i \notin \{i', \poptm_j\}$, we have for any $i\notin \{i', \poptm_j\}$:
\begin{align*}
    p(i',j) \cdot r(i',j) &\leq \eta \hat{p}(i',j) \cdot r(i',j) = \eta \hat{p}(i',j) \cdot 1/\errbound^2 = \hat{p}(i',j) \frac{\eta}{\errbound^2} \leq \hat{p}(i',j)\cdot 1/\errbound\\
    &\leq  \hat{p}(i,j) \cdot r(i,j) \cdot 1/\errbound \leq p(i,j) \cdot r(i,j) \cdot \frac{\eta}{\errbound} \leq p(i,j) \cdot r(i,j),
\end{align*}
Where the first inequality is by the definition of $\eta$, the second inequality is due to $\eta \leq \errbound$, the third inequality is since $r(i,j) \geq 1$ by Observation~\ref{obs:rbiggerthan1}, and the forth inequality is again by the definition of $\eta$. We also have:
\[p(i'j) \cdot r(i',j) = \hat{p}(i,j) \cdot \frac{1}{\errbound^2} < \hat{p}(\poptm_j,j) \cdot \frac{1}{\errbound^2} = p(\poptm_j,j) \cdot r(\poptm_j, j)\]

Since  $p(i',j) \cdot r(i',j) \leq p(i,j) \cdot r(i,j)$ and $p(i'j) \cdot r(i',j)<p(\poptm_j,j) \cdot r(\poptm_j, j) $, by line~\ref{eline:tie}, we have that $x(i',j) = 1$ in \mecherr. Following a similar argument we get for any job $j$ such that $j \notin \cup_{i \in N}J_i$, we have $x(\poptm_j, j) = 1$ in \mecherr.
\end{proof}

Given that $J_i = M_i \setminus \hat{M}_i$, we have Lemma~\ref{lem:upperboundvalforsi} holds true. The approximation obtained by \mecherr \ when the error is within the error tolerance is $(2+\tradeoff)\alpha\eta^2$. The analysis of this approximation is similar to the analysis of the consistency of \mech.

\begin{lemma}
\label{lem:errcons}
Given an error tolerance bound $\errbound$, and an actual error $\eta$, the approximation ratio of mechanism \mecherr\ is $(2+\tradeoff)\alpha\eta^2$ if $\eta \leq \errbound$.
\end{lemma}
\begin{proof}
By Lemma~\ref{lem:errorJ}, when $\eta(\instance, \hat{\instance}) \leq \errbound$, the assignment $\schedule$  returned by the mechanism over input $(\instance, \hat{\instance}, \algo, \tradeoff, \errbound)$ is identical to the assignment returned by $\mech$ over input $(\instance, \instance, \algo, \tradeoff)$. By Lemma~\ref{lem:upperboundvalforsi} we have for any machine $i$
\begin{align*}
    \sum_{j \in M_i} \hat{p}(i,j) \cdot x(i,j) =& 
    \sum_{j \in M_i\cap \hat{M}_i }\hat{p}(i,j) + \sum_{j \in M_i \setminus \hat{M}_i} \hat{p}(i,j)\\
    =& 
    \sum_{j \in M_i\cap \hat{M}_i }\hat{p}(i,j) + \sum_{j \in J_i} \hat{p}(i,j)\tag{by Lemma~\ref{lem:errorJ}}\\
    \leq & \ms(\hat{\instance}, \hat{\schedule}) + (1+\tradeoff)\ms(\hat{\instance}, \hat{\schedule})\tag{by Lemma~\ref{lem:upperboundvalforsi}}\\
    \leq &
    (2+\tradeoff)\ms(\hat{\instance}, \hat{\schedule})\\
    \leq& (2+\tradeoff)\alpha\ms(\hat{\instance}, \hat{\schedule}^\star).\numberthis \label{eq:optofpredicted}
\end{align*}
where $\hat{\schedule}^\star$ is optimal schedule for the predicted instance and the last inequality is by the fact that \algo\ is a $\alpha$-approximation algorithm. Let $\schedule^\star$ be the optimal schedule for the the actual instance $\instance$, we have:
\[ \ms(\hat{\instance},\hat{\schedule}^\star) \leq \ms( \hat{\instance}, \schedule^{\star}).\] 
Now again given the definition of error $\eta$ we have:
\begin{align}\label{eq:actualerror}
    \ms(\hat{\instance},\hat{\schedule}^\star) \leq \ms(\hat{\instance}, \schedule^\star) = \max_{i}\sum_{j \in M} \hat{p}(i,j) \cdot x^\star(i,j) \leq \max_{i}\sum_{j \in M} \eta \cdot p(i,j) \cdot x^\star(i,j) 
    = \eta \cdot \OPT.
\end{align}
Combining Inequalities \eqref{eq:optofpredicted} and \eqref{eq:actualerror} we get:
\[ \max_i\sum_{j \in M}p(i,j) \cdot x(i,j) \leq \max_i\sum_{j \in M} \eta \hat{p}(i,j) \cdot x(i,j) \leq (2+\tradeoff)\alpha\eta\ms(\hat{\instance},\hat{\schedule}^\star) \leq (2+\tradeoff)\alpha\eta^2 \OPT.  \qedhere \]
\end{proof}

The approximation obtained by \mecherr \ when the error is not within the error tolerance is $(1+\frac{1}{\tradeoff})\errbound^2n$. The analysis of this approximation is similar to the analysis of the robustness of \mech.

\begin{lemma}
\label{lem:errrobust}
Given an error tolerance bound $\errbound$, the approximation ratio of mechanism \mecherr\ is  $(1+\frac{1}{\tradeoff})\errbound^2n$.
\end{lemma}
\begin{proof}
 We again partition the set of jobs $M$ into two subset, let $S^*$ be the set of jobs $j$ such that $r(\optm_j, j)\leq 1$ and $M \setminus S^*$ be the rest of the jobs $j$ with $r(\optm_j,j )>1$. Since \mecherr\ only modify some scalar to $\frac{1}{\errbound^2}$, combining with Observation~\ref{obs:rbiggerthan1} we have: $r(i,j) \geq \frac{1}{\errbound^2}$ for any $i,j$. Next,
\begin{align}\label{eq:m1processingtime1}
    \frac{1}{\errbound^2} p(i_j, j ) \leq r(i,j) p(i_j, j) \leq r(\optm_j, j) p(\optm_j, j) \leq p(\optm_j, j),
\end{align}
where the second inequality is by the fact that $x(i_j, j) = 1$, and the last equality is by definition of $S^*$.  We get that
\begin{align}\label{eq:m1sum1}
    \sum_{j \in S^*} p(i_j, j)  \leq  \errbound^2 \sum_{j \in S^*}p(\optm_j, j)  \leq \errbound^2 \sum_{j \in M}p(\optm_j, j) \leq \errbound^2 n \cdot \OPT.
\end{align}
where the first inequality is by Equation~\eqref{eq:m1processingtime1} and the last inequality is by Inequality~\eqref{eq:optlowerbound} from Lemma~\ref{lem:robustness}.
Now consider any job $j \in M \setminus S^*$. Similarly as the analysis of $S^*$ we have 
\begin{equation*}
    \frac{1}{\errbound^2} p(i_j, j ) \leq r(i,j) p(i_j, j) \leq r(\optm_j, j) p(\optm_j, j).
\end{equation*}
Next, we have:
\begin{align}\label{eq:m2sum1}
    \sum_{j \in M \setminus S^*} p(i_j,j) & \leq \errbound^2 \sum_{j \in M \setminus S^*} r(\optm_j, j)p(\optm_j, j) \leq \errbound^2 \frac{n}{\tradeoff} \cdot \OPT,
\end{align}
where  the second inequality is by Inequality \eqref{eq:m2sum}. Combining \eqref{eq:m1sum1} and \eqref{eq:m2sum1} we get:
$$
\max_i \sum_{j \in M_i} p(i, j) \leq \sum_{j \in M} p(i_j, j) \leq \sum_{j \in S^*} p(i_j, j) + \sum_{j \in M \setminus S^*} p(i_j, j) \leq \left(1+\frac{1}{\tradeoff}\right)\errbound^2n \cdot \OPT. \eqno \qedhere$$
\end{proof}

The proof for the strategyproofness \mecherr \ is identical as the proof for the strategyproofness of \mech (Lemma~\ref{lem:astrategyproof}). Together with Lemma~\ref{lem:errcons} and Lemma~\ref{lem:errrobust}, we get the main result for \mecherr.

\begin{theorem}
For any $\tradeoff \in \left(0, \frac{n}{2}-1\right)$ and  error tolerance bound $\errbound > 0$, \mecherr\ is a strategyproof mechanism  that achieves an approximation of $(2+\tradeoff)\alpha\eta^2$ if $\eta \leq \errbound$ and $\left(1+\frac{1}{\tradeoff}\right)\errbound^2n$ otherwise, where $\err$ is the prediction error.
\end{theorem}

Thus, with any constants $\tradeoff, \errbound$, and with \algo \ being the $2$-approximation scheduling algorithm from~\cite{LST90}, we obtain a polynomial-time mechanism that achieves an $O(1)$-approximation when the prediction error $\eta$ is such that  $\eta \leq \errbound$, and an $O(n)$-approximation otherwise.

\section{Perfect Consistency Implies Unbounded Robustness}\label{sec:lower_bound}
Recall that every strategyproof scheduling mechanism must be monotone. Below is a very useful lemma that comes immediately from the monotonicity property and is used in most of the lower bound proofs.

\begin{lemma}[\citep{CKV07}]\label{lem:monotonicityapplication}
Consider instances $\instance$ and $\instance'$ and let $\schedule$ and $\schedule'$ be the allocation produced by a strategyproof scheduling mechanism for the given two instances, respectively. If $\instance$ and $\hat{\instance}$ only differs in the processing time of machine $i$ in such a way that $p'(i,j)> p(i,j)$ when $x(i,j) = 0$, and $p'(i,j)< p(i,j)$ when $x(i,j) = 1$, then $\schedule' = \schedule$.
\end{lemma}

\begin{theorem}
No deterministic strategyproof scheduling mechanism with 1-consistency can achieve any bounded robustness.
\end{theorem}

\begin{proof}
Consider the the predicted instance $\hat{\instance}$ instance with 2 machines and 2 jobs, and the processing time as in the table on the left of Figure~\ref{fig:1consistenctinstance}. Consider an actual instance that is the same as the prediction, i.e., $\instance = \hat{\instance}$. 

\begin{figure}[h]
    \centering
    \begin{tabular}{c|cc}
    m$\backslash$j & 1 & 2\\
\hline
1 & $K^{\star}$ & 1\\
2 & $\infty$ & $K^{\star}$\\
\end{tabular}
\quad
$\rightarrow$
\quad
\begin{tabular}{c|cc}
m$\backslash$j & 1 & 2\\
\hline
1 & $0^{\star}$ & $1+\epsilon$\\
2 & $\infty$ &  $K^{\star}$
\end{tabular}
    \caption{The predicted (on the left)  and actual (on the right) instances showing any deterministic 1-consistent mechanism suffers unbounded robustness.}
    \label{fig:1consistenctinstance}
\end{figure}

By 1-consistency, the mechanism needs to output $x_{11} = x_{22} = 1$ (the allocation is marked with $\star$).
Now consider another instance $\instance'$ represent by table on the right of Figure~\ref{fig:1consistenctinstance} and consider machine $1$, since $p'(1,1)< p(1,1)$and $x(1,1) = 1$ and $p'(1,2)> p(1,2)$ and $x(1,2) = 1$, by Lemma~\ref{lem:monotonicityapplication}, any truthful mechanism needs to output $x'(1,1) = x(1,1) = 1$ and $x'(1,2) = x(1,2) = 0$, resulting a makespan at lest $K$. Since the optimal makespan of $\instance'$ is $1+\epsilon$ by assigning both jobs to machine $1$, it means the robustness of any strategyproof mechanism is at most $K$. Since $K$ can be arbitrarily large, the robustness is therefore unbounded.
\end{proof}

\newpage
\appendix
\section{Missing Analysis from Section~\ref{sec:warmup}}
\label{sec:appwarmuplemma}

\lemrobustnesslowerbound*

\begin{proof}
We first show that the robustness is at most $n^2$. Consider an instance $\instance$, an optimal assignment $\schedule^{\star}$ for $\instance$,  a job $j$, and a machine $i$ such that $j \in M_i$. Let $\optm_j$ be the machine that $j$ is assigned to according to $\schedule^{\star}$. By the mechanism, we have $1 \leq r(i',j) \leq n$ for all machines $i'$. We therefore have:
\[p(i,j) \leq r(i,j) \cdot p(i,j) \leq r(\optm_j,j) \cdot p(\optm_j,j) \leq n \cdot p(\optm_j,j),\]
where  the second inequality is by line~\ref{mech1:tie} and the fact that $j$ is assigned to $i$. We get that 
\[\ms(\instance, \schedule) = \max_{i \in N}\sum_{j \in M_i}p(i,j) \leq \sum_{i \in N}\sum_{j \in M_i}p(i,j)    \leq n \sum_{i \in N}\sum_{j \in M_i} p(\optm_j,j) \leq n^2\max_{i \in N}\sum_{j \in M_i} p(\optm_j,j) = n^2 \OPT.\]
Consider the predicted instance in Figure~\ref{fig:lowerboundinstance} (a).

\begin{figure}[h]
\centering

\begin{tabular}{c|cccccccc}
m$\backslash$j & 1 & 2 & \dots & $n-1$ & $n$ & $n+1$& \dots & $2n-2$\\
\hline
1 & $1$ & &&& $n(n-1)$ \\
2 &  & 1&&&&$n(n-1)$\\
\dots & & & 1&&&&$n(n-1)$ \\
$n-1$ & & & & 1 &&&&$n(n-1)$\\
$n$ & $n$ & $n$ & $n$ & $n$ \\
\end{tabular}
\small (a)

 \vspace{\floatsep}
\begin{tabular}{c|cccccccc}
m$\backslash$j & 1 & 2 & \dots & $n-1$ & $n$ & $n+1$& \dots & $2n-2$\\
\hline
1 & $1+\epsilon$ & &&& $0$&$0$&$0$& $0$ \\
2 &  & $1+\epsilon$&&&$0$&$0$&$0$& $0$\\
\dots & & & $1+\epsilon$&&$0$&$0$&$0$& $0$ \\
$n-1$ & & & & $1+\epsilon$ &$0$&$0$&$0$& $0$\\
$n$ & $n$ & $n$ & $n$ & $n$ \\
\end{tabular}
\small (b)
\caption{(a) Predicted instance $\hat{\instance}$, empty entry means $\hat{p}(i,j) = \infty$. (b) Actual instance $\instance$, empty entry means $p(i,j) = \infty$.}\label{fig:lowerboundinstance}
\end{figure}

Assume that $\algo$ computes the optimal makespan for the predicted instance. We have $\hat{x}(n,i) = 1$ for $i \in [1,n-1]$. Then the mechanism would set $r(i,i) = \frac{1}{n}$ for any $i \in [1,n-1]$. Now consider the actual instance in Figure~\ref{fig:lowerboundinstance}(b), note that for any $i \in [1, n-1]$, we have:
\[r(i,i) \cdot p(i,i) = n \cdot (1 + \epsilon) > n = r(n,i) \cdot r (n,i).\]
Since machine $n$ has a lower scaled processing time for job $i \in [1,n]$, the mechanism would assign all such jobs to machine $n$, leading to a makespan of $n(n-1)$, whereas the optimal makespan is $1+\epsilon$.\qedhere
\end{proof}

\newpage
\bibliographystyle{abbrvnat}
\bibliography{biblio}

\begin{thebibliography}{32}
\providecommand{\natexlab}[1]{#1}
\providecommand{\url}[1]{\texttt{#1}}
\expandafter\ifx\csname urlstyle\endcsname\relax
  \providecommand{\doi}[1]{doi: #1}\else
  \providecommand{\doi}{doi: \begingroup \urlstyle{rm}\Url}\fi

\bibitem[Agrawal et~al.(2022)Agrawal, Balkanski, Gkatzelis, Ou, and
  Tan]{ABGOT22}
P.~Agrawal, E.~Balkanski, V.~Gkatzelis, T.~Ou, and X.~Tan.
\newblock Learning-augmented mechanism design: Leveraging predictions for
  facility location.
\newblock In D.~M. Pennock, I.~Segal, and S.~Seuken, editors, \emph{{EC} '22:
  The 23rd {ACM} Conference on Economics and Computation, Boulder, CO, USA,
  July 11 - 15, 2022}, pages 497--528. {ACM}, 2022.

\bibitem[Antoniadis et~al.(2020)Antoniadis, Gouleakis, Kleer, and
  Kolev]{AGKK20}
A.~Antoniadis, T.~Gouleakis, P.~Kleer, and P.~Kolev.
\newblock Secretary and online matching problems with machine learned advice.
\newblock In H.~Larochelle, M.~Ranzato, R.~Hadsell, M.~F. Balcan, and H.~Lin,
  editors, \emph{Advances in Neural Information Processing Systems}, pages
  7933--7944, 2020.

\bibitem[Ashlagi et~al.(2012)Ashlagi, Dobzinski, and Lavi]{ADL12}
I.~Ashlagi, S.~Dobzinski, and R.~Lavi.
\newblock Optimal lower bounds for anonymous scheduling mechanisms.
\newblock \emph{Math. Oper. Res.}, 37\penalty0 (2):\penalty0 244--258, 2012.

\bibitem[Azar et~al.(2022)Azar, Panigrahi, and Touitou]{azar2022online}
Y.~Azar, D.~Panigrahi, and N.~Touitou.
\newblock Online graph algorithms with predictions.
\newblock \emph{Proceedings of the Thirty-Third Annual ACM-SIAM Symposium on
  Discrete Algorithms}, 2022.

\bibitem[Balkanski et~al.(2022)Balkanski, Ou, Stein, and Wei]{BOSW22}
E.~Balkanski, T.~Ou, C.~Stein, and H.~Wei.
\newblock Scheduling with speed predictions.
\newblock \emph{CoRR}, abs/2205.01247, 2022.

\bibitem[Bamas et~al.(2020{\natexlab{a}})Bamas, Maggiori, Rohwedder, and
  Svensson]{BMRS20}
{\'{E}}.~Bamas, A.~Maggiori, L.~Rohwedder, and O.~Svensson.
\newblock Learning augmented energy minimization via speed scaling.
\newblock In H.~Larochelle, M.~Ranzato, R.~Hadsell, M.~Balcan, and H.~Lin,
  editors, \emph{Advances in Neural Information Processing Systems 33: Annual
  Conference on Neural Information Processing Systems 2020, NeurIPS 2020,
  December 6-12, 2020, virtual}, 2020{\natexlab{a}}.

\bibitem[Bamas et~al.(2020{\natexlab{b}})Bamas, Maggiori, and Svensson]{BMS20}
E.~Bamas, A.~Maggiori, and O.~Svensson.
\newblock The primal-dual method for learning augmented algorithms.
\newblock In H.~Larochelle, M.~Ranzato, R.~Hadsell, M.~F. Balcan, and H.~Lin,
  editors, \emph{Advances in Neural Information Processing Systems}, pages
  20083--20094, 2020{\natexlab{b}}.

\bibitem[Bampis et~al.(2022)Bampis, Dogeas, Kononov, Lucarelli, and
  Pascual]{BDKLP22}
E.~Bampis, K.~Dogeas, A.~V. Kononov, G.~Lucarelli, and F.~Pascual.
\newblock Scheduling with untrusted predictions.
\newblock In L.~D. Raedt, editor, \emph{Proceedings of the Thirty-First
  International Joint Conference on Artificial Intelligence, {IJCAI} 2022,
  Vienna, Austria, 23-29 July 2022}, pages 4581--4587. ijcai.org, 2022.

\bibitem[Banerjee et~al.(2022)Banerjee, Gkatzelis, Gorokh, and
  Jin]{banerjee2020online}
S.~Banerjee, V.~Gkatzelis, A.~Gorokh, and B.~Jin.
\newblock Online nash social welfare maximization with predictions.
\newblock In \emph{Proceedings of the 2022 {ACM-SIAM} Symposium on Discrete
  Algorithms, {SODA} 2022}. {SIAM}, 2022.

\bibitem[Christodoulou et~al.()Christodoulou, Koutsoupias, and Vidali]{CKV07}
G.~Christodoulou, E.~Koutsoupias, and A.~Vidali.
\newblock A lower bound for scheduling mechanisms.
\newblock In N.~Bansal, K.~Pruhs, and C.~Stein, editors, \emph{Proceedings of
  the Eighteenth Annual {ACM-SIAM} Symposium on Discrete Algorithms, {SODA}
  2007, New Orleans, Louisiana, USA, January 7-9, 2007}.

\bibitem[Christodoulou et~al.(2009)Christodoulou, Koutsoupias, and
  Vidali]{CKV09}
G.~Christodoulou, E.~Koutsoupias, and A.~Vidali.
\newblock A lower bound for scheduling mechanisms.
\newblock \emph{Algorithmica}, 55\penalty0 (4):\penalty0 729--740, 2009.

\bibitem[Christodoulou et~al.(2010)Christodoulou, Koutsoupias, and
  Kov{\'{a}}cs]{CKK10}
G.~Christodoulou, E.~Koutsoupias, and A.~Kov{\'{a}}cs.
\newblock Mechanism design for fractional scheduling on unrelated machines.
\newblock \emph{{ACM} Trans. Algorithms}, 6\penalty0 (2):\penalty0 38:1--38:18,
  2010.

\bibitem[Christodoulou et~al.(2021)Christodoulou, Koutsoupias, and
  Kov{\'{a}}cs]{CKK21}
G.~Christodoulou, E.~Koutsoupias, and A.~Kov{\'{a}}cs.
\newblock On the nisan-ronen conjecture.
\newblock In \emph{62nd {IEEE} Annual Symposium on Foundations of Computer
  Science, {FOCS} 2021, Denver, CO, USA, February 7-10, 2022}, pages 839--850.
  {IEEE}, 2021.

\bibitem[Dobzinski and Shaulker(2020)]{DS20}
S.~Dobzinski and A.~Shaulker.
\newblock Improved lower bounds for truthful scheduling.
\newblock \emph{CoRR}, abs/2007.04362, 2020.

\bibitem[D{\"u}tting et~al.(2021)D{\"u}tting, Lattanzi, Paes~Leme, and
  Vassilvitskii]{dutting2021secretaries}
P.~D{\"u}tting, S.~Lattanzi, R.~Paes~Leme, and S.~Vassilvitskii.
\newblock Secretaries with advice.
\newblock In \emph{Proceedings of the 22nd ACM Conference on Economics and
  Computation}, pages 409--429, 2021.

\bibitem[Giannakopoulos et~al.(2020)Giannakopoulos, Hammerl, and
  Po{\c{c}}as]{GH20}
Y.~Giannakopoulos, A.~Hammerl, and D.~Po{\c{c}}as.
\newblock A new lower bound for deterministic truthful scheduling.
\newblock In T.~Harks and M.~Klimm, editors, \emph{Algorithmic Game Theory -
  13th International Symposium, {SAGT} 2020, Augsburg, Germany, September
  16-18, 2020, Proceedings}, volume 12283 of \emph{Lecture Notes in Computer
  Science}, pages 226--240. Springer, 2020.

\bibitem[Gkatzelis et~al.(2022)Gkatzelis, Kollias, Sgouritsa, and Tan]{GKST22}
V.~Gkatzelis, K.~Kollias, A.~Sgouritsa, and X.~Tan.
\newblock Improved price of anarchy via predictions.
\newblock In D.~M. Pennock, I.~Segal, and S.~Seuken, editors, \emph{{EC} '22:
  The 23rd {ACM} Conference on Economics and Computation, Boulder, CO, USA,
  July 11 - 15, 2022}, pages 529--557. {ACM}, 2022.

\bibitem[Im et~al.(2021)Im, Kumar, Montazer~Qaem, and Purohit]{im2021online}
S.~Im, R.~Kumar, M.~Montazer~Qaem, and M.~Purohit.
\newblock Online knapsack with frequency predictions.
\newblock \emph{Advances in Neural Information Processing Systems}, 34, 2021.

\bibitem[Koutsoupias and Vidali(2007)]{KV07}
E.~Koutsoupias and A.~Vidali.
\newblock A lower bound of 1+\emph{phi} for truthful scheduling mechanisms.
\newblock In L.~Kucera and A.~Kucera, editors, \emph{Mathematical Foundations
  of Computer Science 2007, 32nd International Symposium, {MFCS} 2007,
  Cesk{\'{y}} Krumlov, Czech Republic, August 26-31, 2007, Proceedings}, volume
  4708 of \emph{Lecture Notes in Computer Science}, pages 454--464. Springer,
  2007.

\bibitem[Lattanzi et~al.(2020)Lattanzi, Lavastida, Moseley, and
  Vassilvitskii]{LLMV20}
S.~Lattanzi, T.~Lavastida, B.~Moseley, and S.~Vassilvitskii.
\newblock Online scheduling via learned weights.
\newblock In S.~Chawla, editor, \emph{Proceedings of the 2020 {ACM-SIAM}
  Symposium on Discrete Algorithms, {SODA} 2020, Salt Lake City, UT, USA,
  January 5-8, 2020}, pages 1859--1877. {SIAM}, 2020.

\bibitem[Lenstra et~al.(1990)Lenstra, Shmoys, and Tardos]{LST90}
J.~K. Lenstra, D.~B. Shmoys, and {\'{E}}.~Tardos.
\newblock Approximation algorithms for scheduling unrelated parallel machines.
\newblock \emph{Math. Program.}, 46:\penalty0 259--271, 1990.

\bibitem[Li and Xian(2021)]{LX21}
S.~Li and J.~Xian.
\newblock Online unrelated machine load balancing with predictions revisited.
\newblock In M.~Meila and T.~Zhang, editors, \emph{Proceedings of the 38th
  International Conference on Machine Learning, {ICML} 2021, 18-24 July 2021,
  Virtual Event}, volume 139 of \emph{Proceedings of Machine Learning
  Research}, pages 6523--6532. {PMLR}, 2021.

\bibitem[Lindermayr and Megow()]{alps}
A.~Lindermayr and N.~Megow.
\newblock Alps.
\newblock URL \url{https://algorithms-with-predictions.github.io/}.

\bibitem[Lykouris and Vassilvtiskii(2018)]{lykouris2018competitive}
T.~Lykouris and S.~Vassilvtiskii.
\newblock Competitive caching with machine learned advice.
\newblock In \emph{International Conference on Machine Learning}, pages
  3296--3305. PMLR, 2018.

\bibitem[Mitzenmacher(2020)]{Mit}
M.~Mitzenmacher.
\newblock Scheduling with predictions and the price of misprediction.
\newblock In T.~Vidick, editor, \emph{11th Innovations in Theoretical Computer
  Science Conference, {ITCS} 2020, January 12-14, 2020, Seattle, Washington,
  {USA}}, volume 151 of \emph{LIPIcs}, pages 14:1--14:18. Schloss Dagstuhl -
  Leibniz-Zentrum f{\"{u}}r Informatik, 2020.

\bibitem[Mitzenmacher and Vassilvitskii(2021)]{MV21}
M.~Mitzenmacher and S.~Vassilvitskii.
\newblock \emph{Algorithms with Predictions}, page 646–662.
\newblock Cambridge University Press, 2021.
\newblock \doi{10.1017/9781108637435.037}.

\bibitem[Nisan and Ronen(1999)]{NR99}
N.~Nisan and A.~Ronen.
\newblock Algorithmic mechanism design (extended abstract).
\newblock In J.~S. Vitter, L.~L. Larmore, and F.~T. Leighton, editors,
  \emph{Proceedings of the Thirty-First Annual {ACM} Symposium on Theory of
  Computing, May 1-4, 1999, Atlanta, Georgia, {USA}}, pages 129--140. {ACM},
  1999.

\bibitem[Nisan and Ronen(2001)]{NR01}
N.~Nisan and A.~Ronen.
\newblock Algorithmic mechanism design.
\newblock \emph{Games Econ. Behav.}, 35\penalty0 (1-2):\penalty0 166--196,
  2001.

\bibitem[Purohit et~al.(2018{\natexlab{a}})Purohit, Svitkina, and Kumar]{KPZ18}
M.~Purohit, Z.~Svitkina, and R.~Kumar.
\newblock Improving online algorithms via ml predictions.
\newblock In S.~Bengio, H.~Wallach, H.~Larochelle, K.~Grauman, N.~Cesa-Bianchi,
  and R.~Garnett, editors, \emph{Advances in Neural Information Processing
  Systems}. Curran Associates, Inc., 2018{\natexlab{a}}.

\bibitem[Purohit et~al.(2018{\natexlab{b}})Purohit, Svitkina, and Kumar]{PSK18}
M.~Purohit, Z.~Svitkina, and R.~Kumar.
\newblock Improving online algorithms via {ML} predictions.
\newblock In S.~Bengio, H.~M. Wallach, H.~Larochelle, K.~Grauman,
  N.~Cesa{-}Bianchi, and R.~Garnett, editors, \emph{Advances in Neural
  Information Processing Systems 31: Annual Conference on Neural Information
  Processing Systems 2018, NeurIPS 2018, December 3-8, 2018, Montr{\'{e}}al,
  Canada}, pages 9684--9693, 2018{\natexlab{b}}.

\bibitem[Saks and Yu(2005)]{SY05}
M.~E. Saks and L.~Yu.
\newblock Weak monotonicity suffices for truthfulness on convex domains.
\newblock In J.~Riedl, M.~J. Kearns, and M.~K. Reiter, editors,
  \emph{Proceedings 6th {ACM} Conference on Electronic Commerce (EC-2005),
  Vancouver, BC, Canada, June 5-8, 2005}, pages 286--293. {ACM}, 2005.

\bibitem[Xu and Lu(2022)]{XL22}
C.~Xu and P.~Lu.
\newblock Mechanism design with predictions.
\newblock In L.~D. Raedt, editor, \emph{Proceedings of the Thirty-First
  International Joint Conference on Artificial Intelligence, {IJCAI} 2022,
  Vienna, Austria, 23-29 July 2022}, pages 571--577. ijcai.org, 2022.
\newblock \doi{10.24963/ijcai.2022/81}.
\newblock URL \url{https://doi.org/10.24963/ijcai.2022/81}.

\end{thebibliography}

\end{document}